\documentclass[pra,twocolumn,amsmath,amssymb,superscriptaddress,notitlepage]{revtex4-2}
     
    \usepackage{graphicx}
    \usepackage{mathtools,physics}
    \usepackage{amsthm}
     
    \theoremstyle{plain}
    \newtheorem{lemma}{Lemma}
    \newtheorem{corollary}{Corollary}
     
    \usepackage[utf8]{inputenc}
     
    \usepackage{hyperref}
    \usepackage{enumitem}
     
    \hypersetup{
    	colorlinks   = true, 
    	urlcolor     = blue, 
    	linkcolor    = blue, 
    	citecolor   = blue 
    }
    \usepackage[capitalize]{cleveref}
     
    \crefname{appendix}{Appendix}{Appendices}
     
    \makeatletter
    \AddToHook{cmd/appendix/before}{\def\cref@section@alias{appendix}\def\cref@subsection@alias{appendix}}
    \makeatother
     
    \usepackage[english]{babel}
    \usepackage[autostyle, english = american]{csquotes} \MakeOuterQuote{"}
    \usepackage{bm}
     
    \usepackage[dvipsnames]{xcolor}

    \newcommand{\eph}{e_\mathrm{ph}}
     
    \newcommand{\affvqcc}{Vigo Quantum Communication Center, University of Vigo, Vigo E-{36310}, Spain}
    \newcommand{\affuvigo}{Escuela de Ingeniería de Telecomunicación, Department of Signal Theory and Communications, University of Vigo, Vigo E-36310, Spain}
    \newcommand{\affatlantic}{atlanTTic Research Center, University of Vigo, Vigo E-36310, Spain}
    \newcommand{\afftoyama}{Faculty of Engineering, University of Toyama, Gofuku 3190, Toyama 930-8555, Japan}
     
    \makeatletter 
     
    \renewcommand\onecolumngrid{
    	\do@columngrid{one}{\@ne}%
    	\def\set@footnotewidth{\onecolumngrid}
    	\def\footnoterule{\kern-6pt\hrule width 1.5in\kern6pt}%
    }
     
    \renewcommand\twocolumngrid{
    	\def\footnoterule{
    		\dimen@\skip\footins\divide\dimen@\thr@@
    		\kern-\dimen@\hrule width.5in\kern\dimen@}
    	\do@columngrid{mlt}{\tw@}
    }%
     
    \makeatother    

\makeatletter
\AtBeginDocument{\let\LS@rot\@undefined}
\makeatother

\usepackage{pdfpages}
\usepackage{pgffor}

\def\supplementfilename{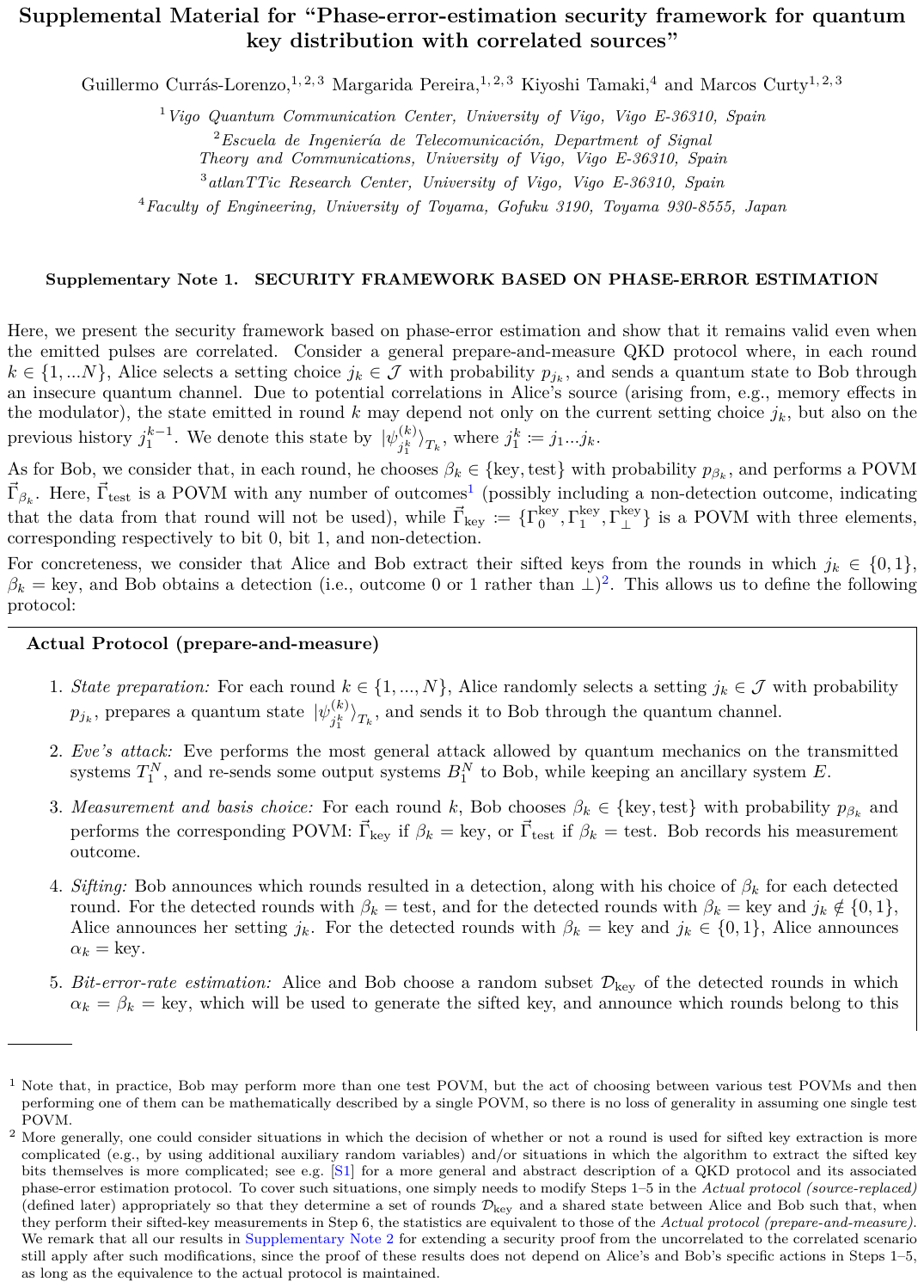}
\pdfximage{\supplementfilename}
\edef\numbersupplementpages{\the\pdflastximagepages}

    \newif\ifarXiv
    \arXivtrue

\begin{document}

        	\setlength{\parskip}{3pt}
        	\setlength{\parindent}{0pt}
        	
        	\author{Guillermo Currás-Lorenzo}
        	\author{Margarida Pereira}
        	\affiliation{\affvqcc} \affiliation{\affuvigo} \affiliation{\affatlantic} 
        	\author{Kiyoshi Tamaki}
        	\affiliation{\afftoyama} 
        	\author{Marcos Curty}
        	\affiliation{\affvqcc} \affiliation{\affuvigo} \affiliation{\affatlantic}
        	
        	\title{Phase-error-estimation security framework for quantum key distribution with correlated sources}
        	
        	\setlength{\parskip}{3pt}
        	\setlength{\parindent}{0pt}
        	
        	\begin{abstract}
        		Practical QKD modulators introduce correlations between consecutively emitted pulses due to bandwidth limitations, violating key assumptions underlying many security proof techniques. Here, we address this problem by introducing a simple yet powerful mathematical framework to directly extend phase-error-estimation-based security proofs for imperfect but uncorrelated sources to also incorporate encoding correlations. Our framework  overcomes important limitations of previous approaches in terms of generality and rigor, significantly narrowing the gap between theoretical security guarantees and real-world QKD implementations.
        		
        	\end{abstract}
        	\maketitle
        	
        	\textit{Introduction.}---Quantum key distribution (QKD) promises information-theoretically secure communications by exploiting fundamental quantum mechanical principles. However, a central challenge in practical QKD is rigorously accounting for device imperfections that inevitably arise in real systems. While security proofs have been developed to handle various imperfections~\cite{gottesmanSecurityQuantum2004,fungSecurityProof2009,maroySecurityQuantum2010,tamakiLosstolerantQuantum2014,pereiraQuantumKey2020,curras-lorenzoSecurityFramework2025,tupkaryPhaseError2025,sixtoQuantumKey2025,kaminRenyiSecurity2025,marwahProvingSecurity2024,curras-lorenzoSecurityQuantum2026}, encoding correlations---which arise naturally from the limited bandwidth of optical modulators \cite{grunenfelderPerformanceSecurity2020,agulleiroModelingCharacterization2025} and cause each emitted quantum state to depend on setting choices from previous rounds---remain particularly challenging to incorporate within existing security frameworks. These correlations fundamentally violate key assumptions underlying several popular security proof techniques such as the Postselection Technique~\cite{christandlPostselectionTechnique2009,naharPostselectionTechnique2024}, Quantum de Finetti approaches~\cite{rennerSymmetryLarge2007} and the Marginal-constrained Entropy Accumulation Theorem (MEAT)~\cite{arqandMarginalconstrainedEntropy2025,kaminRenyiSecurity2025}, invalidating their application to realistic QKD setups that inevitably suffer from such correlations \footnote{The Postselection Technique~\cite{christandlPostselectionTechnique2009,naharPostselectionTechnique2024} and approaches based on the Quantum de Finetti theorem~\cite{rennerSymmetryLarge2007} require the global protocol state to be permutation-invariant, enabling a reduction from general attacks to collective attacks. However, with encoding correlations, the temporal ordering of settings affects the physical state, breaking this symmetry. Similarly, the Marginal-constrained Entropy Accumulation Theorem (MEAT)~\cite{arqandMarginalconstrainedEntropy2025} models the protocol state as produced by a sequence of channels, each acting on an input state satisfying a fixed marginal constraint. With correlations, Alice's source state in round $k$ depends on all previous settings $j_1^{k-1}$, violating the required factorization structure of the source-replacement state.}.
        	
        	
        	On the other hand, security proofs based on phase-error estimation---including both proofs based on entropic uncertainty relations (EUR) with the leftover hashing lemma (LHL)~\cite{tomamichelUncertaintyRelation2011,tomamichelTightFinitekey2012,tomamichelLargelySelfcontained2017} and proofs based on phase-error correction (PEC)~\cite{koashiSimpleSecurity2009}---face no fundamental barriers to incorporating encoding correlations, as we rigorously establish in this work. Nevertheless, such correlations significantly complicate the phase-error estimation task, and were largely overlooked until a key conceptual breakthrough was introduced by Ref.~\cite{pereiraQuantumKey2020}. To illustrate the key idea, consider the simplest case of nearest-neighbor correlations. From the perspective of the $k$-th pulse, these correlations manifest in two ways: (i) the photonic system of round $k$ depends on the setting choice made in round $(k-1)$, and (ii) the setting choice in round $k$ affects the encoding of the photonic pulse in round $(k+1)$. The crucial insight of Ref.~\cite{pereiraQuantumKey2020} is that effect (i) resembles an encoding flaw, while effect (ii) is analogous to information leakage through a side channel system.

        	Building on this insight, Ref.~\cite{pereiraQuantumKey2020} (see also \cite{mizutaniSecurityRoundrobin2021,pereiraModifiedBB842023,curras-lorenzoSecurityFramework2025}) shows how security proofs capable of handling encoding flaws and side-channel leakage can be adapted to incorporate encoding correlations. The approach partitions rounds into $(l_c+1)$ groups according to $k \bmod (l_c+1)$, where $l_c$ is the maximum correlation length, and treats each group as an independent subprotocol. For instance, with nearest-neighbor correlations ($l_c=1$), rounds are split into even and odd groups, and a phase-error rate bound for the odd-rounds key is established by conditioning on fixed values of the even-rounds' settings. Then, these works argue that, since this bound holds for any fixed values of the even-rounds' settings, the resulting security proof for the odd-rounds key remains valid regardless of the value of the even-rounds key. Applying the same argument to the even-rounds key, the security of the full key then follows from composing the two individual proofs.

        	
        	While this approach represents the only known method to incorporate encoding correlations into phase-error-estimation-based proofs to date, it suffers from several limitations that reduce its practical usefulness:
        	\begin{enumerate}[label=(\alph*)]

        		
        		\item  \textit{Privacy amplification complexity}: It requires that privacy amplification is performed separately for the $(l_c+1)$ subkeys. This increases implementation complexity and introduces potential failure points.

        		\item \textit{Composability concerns:} 
        		Despite attempts to formalize the composability arguments needed to combine the security proofs for the individual subprotocols~\cite{mizutaniSecurityRoundrobin2021}, the validity of this composition remains contested, with a recent work explicitly labeling it a ``conjecture''~\cite{kaminRenyiSecurity2025}. This concern is also reflected in a recent review paper that classifies phase-error-estimation-based proofs as robust only against independent device imperfections \cite[Table III]{tupkaryQKDSecurity2025}.
        				
        		\item \textit{Protocol specificity:} Previous works consider specific protocols and security proofs on a case-by-case basis, without providing a fully general mathematical framework to address correlations across protocols.

        	\end{enumerate}
        	
        	In this work, we overcome these challenges by constructing a rigorous and general framework to extend phase-error-estimation-based security proofs to correlated sources in a systematic way, addressing Limitation~(c). In doing so, we rigorously establish that phase-error-rate bounds for individual partitions can be directly combined into a single bound on the overall phase-error rate of the full sifted key. Thus, by applying our framework, one can achieve security through a single privacy amplification step on the full key, eliminating Limitation~(a) and circumventing the composability arguments of Limitation~(b). 
        	
        	For concreteness, we focus our presentation on prepare-and-measure protocols, i.e., protocols in which Alice sends states to Bob. However, our framework can also be applied to address encoding correlations in interference-based protocols, i.e., protocols in which Alice and Bob send states to an untrusted middle node Charlie. To simulate the resulting key rate, we consider, as a particular example, a BB84 phase modulator modeled as a linear time-invariant (LTI) system, although our framework is general and not restricted to any specific device model.

        	\textit{Source replacement scheme with correlated sources.}---
        	Consider a general prepare-and-measure QKD protocol where, in round $k \in \{1,...,N\}$, Alice selects setting $j_k \in \mathcal{J}$ with probability $p_{j_k}$, where $\mathcal{J}$ denotes the alphabet of all her possible setting choices (e.g., in BB84, this includes all combinations of her basis and bit choices, i.e., $\mathcal{J} = \{(Z,0),(Z,1), (X,0), (X,1)\}$). Then, she emits a state $\ket*{\psi^{(k)}_{j_1^k}}_{T_k}$, where $j_1^k \coloneqq j_1 ... j_k$. Due to encoding correlations, this state depends not only on the current setting $j_k$, but also on the history of previous settings  $j_1^{k-1}$. Alice's state preparation is equivalent to generating the global source-replacement state
        	\begin{equation}
        		\label{eq:source_replacement_correlated}
        		\ket*{\Psi_N}_{A_1^N T_1^N} = \sum_{j_1^N} \bigotimes_k \sqrt{p_{j_k}} \ket*{j_k}_{A_k} \ket*{\psi^{(k)}_{j_1^k}}_{T_k},	\end{equation}
        	and then measuring systems $A_1^N \coloneqq A_1 ... A_N$ in the computational basis $\{\ket{j_k}_{A_k}\}_{j_k \in \mathcal{J}}$, while sending systems $T_1^N$ through the channel. 
        	
        	Uncorrelated sources correspond to the special case where $\ket*{\psi^{(k)}_{j_1^k}}_{T_k} \equiv \ket*{\psi^{(k)}_{j_k}}_{T_k}$. In this case, the global state factorizes as
        	\begin{equation}
        		\label{eq:source_replacement_uncorrelated}
        		\ket*{\Psi_N}_{A_1^N T_1^N} =  \bigotimes_k \sum_{j_k} \sqrt{p_{j_k}} \ket*{j_k}_{A_k} \ket*{\psi^{(k)}_{j_k}}_{T_k}.
        	\end{equation}

        	\textit{Phase-error estimation with correlated sources.}---Security proofs based on phase-error estimation follow a specific approach. Using the source-replacement state, one first defines a scenario equivalent to the actual protocol in which Alice and Bob initially determine \textit{which} rounds will be used to generate the sifted key, and only later extract the actual key bits. In this extraction phase, Alice measures qubit systems in the computational basis $\{\ket*{0},\ket*{1}\}$, while Bob performs a two-outcome POVM $\{G_0,G_1\}$ on his sifted-key systems.
        	
        	Then, one considers a fictitious \textit{phase-error estimation protocol} in which Alice instead measures her qubits in the $\{\ket*{+},\ket*{-}\}$ basis, where $\ket*{\pm} = (\ket*{0} \pm \ket*{1})/\sqrt{2}$, and Bob uses his side information to attempt to predict her outcomes. The phase-error rate $\bm{e_\mathrm{ph}}$ is defined as the fraction of incorrect predictions, and bounding this quantity is the key to establishing security of the actual protocol. Specifically, one must prove that
        	\begin{equation}
        		\label{eq:general_eph_bound}
        		\Pr[\bm{e_\mathrm{ph}} > \mathcal{E}_\mathrm{ph} (\bm{\vec n};N,\epsilon)] \leq \epsilon,
        	\end{equation}
        	where $\bm{\vec n}$ is a random vector representing the announced data  (before post-processing), $\epsilon$ is the bound's failure probability, $\mathcal{E}_\mathrm{ph}$ is a function relating these quantities, and the bound is established for any value of the total number of transmitted rounds $N$. Note that we use the convention that bold variables represent random variables.
        	
        	For uncorrelated sources, it is well established that such a bound is enough to guarantee security via either EUR+LHL~\cite{tomamichelUncertaintyRelation2011,tomamichelTightFinitekey2012,tomamichelLargelySelfcontained2017} or by using PEC arguments~\cite{koashiSimpleSecurity2009}, even when the final key is of variable length~\cite{curras-lorenzoTightFinitekey2021,tupkaryPhaseError2025,hayashiConciseTight2012,kawakamiSecurity}. In particular, under the EUR+LHL framework, one obtains a final key of length~\cite{tupkaryPhaseError2025}
        	\begin{equation}
        		\begin{aligned}
        			\bm{l} &= \bm{n_K} \big[1-h\big(\mathcal{E}_\mathrm{ph} (\bm{\vec n};N,\epsilon)\big)\big] - \lambda_\mathrm{EC}(\bm{\vec n}) \\
        			&- 2 \log_2\big(1/(2\varepsilon_\mathrm{PA})\big)-\log_2(2/\varepsilon_\mathrm{EV}),
        		\end{aligned}
        		\label{eq:sklength}
        	\end{equation}
            with security parameter $(\varepsilon_\mathrm{corr}+\varepsilon_\mathrm{sec})$, where
            \begin{equation}
            \varepsilon_\mathrm{corr} = \varepsilon_\mathrm{EV}, ~\text{and}~~\varepsilon_\mathrm{sec} = 2 \sqrt{\epsilon} + \varepsilon_\mathrm{PA}.
            \end{equation}
            Here, $\bm{n_K}$ is the sifted key length, $\varepsilon_\mathrm{EV}$ is the error verification failure probability, $\varepsilon_\mathrm{PA}>0$ is freely chosen, and $\lambda_\mathrm{EC}$ is a function of $\bm{\vec n}$ that determines the number of bits leaked during error correction.
        	
        	In this work, we rigorously establish that a \textit{phase-error estimation protocol} defined for the uncorrelated source-replacement state in \cref{eq:source_replacement_uncorrelated} remains equally valid when considering the correlated source-replacement state in \cref{eq:source_replacement_correlated}. Consequently, proving a bound of the form in \cref{eq:general_eph_bound} suffices to guarantee security even with correlated sources (see the Supplemental Material for the detailed construction and proof). This is significant because it demonstrates that the fundamental security framework---deriving secure key lengths from phase-error-rate bounds via EUR+LHL or PEC---applies naturally to correlated sources. The challenge thus reduces to deriving phase-error rate upper bounds that account for correlations. Our main contribution below addresses this challenge by providing a general method to extend phase-error rate upper bounds from uncorrelated to correlated sources. For the proof of the technical results below, as well as the full statement of our framework, see the Supplemental Material.

        	\begin{corollary}
        		\label{cor:main}
        		
        		Consider a prepare-and-measure QKD protocol with an uncorrelated source. In each round $k$, the source is characterized by a family of states $\{\ket*{\psi_{j_k}^{(k)}}_{T_k}\}_{j_k \in \mathcal{J}}$ indexed by the setting $j_k \in \mathcal{J}$. Suppose there exists an admissibility set $\mathcal{S}$ of families of single-round states such that the phase-error bound in \cref{eq:general_eph_bound} is guaranteed to hold as long as the source satisfies
        		\begin{equation}
        			\{\ket*{\psi^{(k)}_{j_k}}_{T_k}\}_{j_k \in \mathcal{J}} \in \mathcal{S}, \quad \forall k.
        		\end{equation}

        		Now consider the analogous protocol with a source exhibiting correlations up to length $l_c$, where in round $k$ Alice emits a state $\ket*{\psi^{(k)}_{j_{k-l_c}^k}}_{T_k}$ that depends on the setting history up to $l_c$ rounds ago. For any sequence of settings $j_{k-l_c}^{k+l_c}$ (with the convention that indices outside $\{1,...,N\}$ are truncated appropriately) define the joint state emitted in rounds $k$ to $k+l_c$ as
        		\begin{equation}
        			\label{eq:multi_round_correlated_states_def}
        			\ket*{\Psi_{j_{k-l_c}^{k+l_c}}^{(k)}}_{T_k^{k+l_c}} = \bigotimes_{m = k}^{k+l_c} \ket*{\psi_{j_{m-l_c}^m}^{(m)}}_{T_m}.
        		\end{equation}
        		
        		Suppose that, for every round $k$ and every fixed choice of past and future settings $(j_{k-l_c}^{k-1},j_{k+1}^{k+l_c})$, the family $\{\ket*{\Psi_{j_{k-l_c}^{k+l_c}}^{(k)}}_{T_k^{k+l_c}}\}_{j_k \in \mathcal{J}}$ obtained by varying $j_k$ is isometrically equivalent to (i.e., has the same Gram matrix as) an acceptable family of single-round states $\{\ket*{\varphi_{j}^{(j_{k-l_c}^{k-1},j_{k+1}^{k+l_c})}}_{T_k}\}_{j \in \mathcal{J}} \in \mathcal{S}$.

        		Then, partitioning the rounds $k \in \{1,...,N\}$ into $(l_c+1)$ sets $I_w = \{k : k \equiv w \bmod{(l_c+1)}\}$ with $w=0,...,l_c$, the phase-error rate in the correlated scenario satisfies
        		\begin{equation}
        			\label{eq:general_eph_bound_correl}
        			\Pr[\bm\eph > \mathcal{E}'_{\mathrm{ph}}  (\bm{\vec n};N,\epsilon, l_c)]  \leq \epsilon,
        		\end{equation}
                with
                \begin{equation}
                \label{eq:Eprimeph}
                  \mathcal{E}'_{\mathrm{ph}}  (\bm{\vec n};N,\epsilon, l_c) = \frac{\sum_{w=0}^{l_c} \bm{n_K^{(w)}} \mathcal{E}_{\mathrm{ph}}  \big(\bm{\vec n^{(w)}};N^{(w)},\frac{\epsilon}{l_c+1}\big)}{\bm{n_K}},
                \end{equation}
        		where $\bm{\vec n^{(w)}}$ is the restriction of the announced data vector $\bm{\vec n}$ to rounds in $I_w$, $N^{(w)} = \abs{I_w}$, $\bm{n_K^{(w)}}$ is the number of sifted key bits from rounds in $I_w$, and $\bm{n_K}=\sum_{w=0}^{l_c} \bm{n_K^{(w)}}$.
        	\end{corollary}
        	
        	\textit{Interpretation.}---\cref{cor:main} essentially says that, as long as the family of multi-round correlated states $\{\ket*{\Psi_{j_{k-l_c}^{k+l_c}}^{(k)}}_{T_k^{k+l_c}}\}_{j_k \in \mathcal{J}}$ satisfies the single-round conditions of the original uncorrelated proof (up to an isometry), then one can divide the protocol rounds into $(l_c+1)$ partitions, apply the uncorrelated phase-error estimation formula to upper bound the phase-error rate of each partition, and take the weighted average to obtain an upper bound on the overall phase-error rate. Then, the secret-key length of the final key is given by \cref{eq:sklength} after substituting $\mathcal{E}_\mathrm{ph} (\bm{\vec n};N,\epsilon)$ by $\mathcal{E}'_{\mathrm{ph}}  (\bm{\vec n};N,\epsilon, l_c)$.
        
            This implies that the asymptotic secret-key rate obtainable with a correlated source whose multi-round states satisfy the admissibility conditions in $\mathcal{S}$ is the same as that of an uncorrelated source whose single-round states are in $\mathcal{S}$. However, the correlated scenario incurs a larger finite-key penalty, which manifests in two ways: (1) the failure probability entering each application of the uncorrelated phase-error estimator $\mathcal{E}_{\mathrm{ph}}$ is reduced from $\epsilon$ to $\epsilon/(l_c+1)$ (see \cref{eq:Eprimeph}), and (2) the effective block size to which this estimator is applied is reduced from $N$ to $N/(l_c+1)$, since rounds are divided into $(l_c+1)$ partitions. In practice, effect (2) dominates, since finite-key deviation bounds typically depend only logarithmically on the failure probability, making effect (1) comparatively mild. Thus, to achieve similar finite-key performance as in the uncorrelated case, the total block size $N$ should be increased by a factor of roughly $(l_c+1)$.
        	
        	Note that, for some correlation models, the Gram matrix of the family $\{\ket*{\Psi_{j_{k-l_c}^{k+l_c}}^{(k)}}_{T_k^{k+l_c}}\}_{j_k \in \mathcal{J}}$ may depend on $(j_{k-l_c}^{k-1},j_{k+1}^{k+l_c})$ and/or the round $k$. For such models, the admissibility set $\mathcal{S}$ cannot consist of a single family $\{\ket*{\varphi_j}\}_{j \in \mathcal{J}}$, i.e., the original security proof should consider some form of partial state characterization. A common approach, employed in Refs.~\cite{curras-lorenzoSecurityFramework2025,pereiraQuantumKey2019,pereiraQuantumKey2020,navarretePracticalQuantum2021,pereiraModifiedBB842023}, is to require fidelity bounds to reference states, a condition originally developed to handle information leakage through hard-to-characterize degrees of freedom such as mode dependencies or Trojan-horse attacks. In the following corollary, we show how our framework extends such security proofs to also incorporate encoding correlations.
        	For other examples of admissibility sets and protocols to which our results apply, see \cref{app:b92,app:scs}.
        	
        	\begin{corollary}[Fidelity bound to reference states]
        	\label{cor:fidelity bound}
        	Consider a prepare-and-measure QKD protocol with an uncorrelated source, and suppose there exists a set of reference states $\{\ket*{\phi_{j}}\}_{j \in \mathcal{J}}$ such that the phase-error bound in~\cref{eq:general_eph_bound} holds as long as
        	\begin{equation}
        		\abs*{\braket*{\phi_{j_k}}{\psi_{j_k}^{(k)}}_{T_k}}^2 \geq 1-\xi_{j_k}, \quad \forall k, \forall j_k \in \mathcal{J}.
        		\label{eq:epsilon_side}
        	\end{equation}

        	For an analogous protocol with a source with correlations up to length $l_c$, fix, for every round $k$ and every choice of past and future settings $(j_{k-l_c}^{k-1},j_{k+1}^{k+l_c})$, an arbitrary family of states $\{\ket*{\phi_{j_{k-l_c}^{k+l_c}}^{(k)}}_{T_k}\}_{j_k \in \mathcal{J}}$ with the same Gram matrix as the family of reference states $\{\ket*{\phi_{j}}\}_j$, and an arbitrary state $\ket*{\lambda_{j_{k-l_c}^{k-1},j_{k+1}^{k+l_c}}^{(k)}}_{T_{k+1}^{k+l_c}}$ independent of $j_k$, such that
        	\begin{equation}
        		\begin{aligned}
        			\label{eq:correlations_general}
        			&\abs{\bra*{\phi_{j_{k-l_c}^{k+l_c}}^{(k)}}_{T_k} \otimes \bra*{\lambda_{j_{k-l_c}^{k-1},j_{k+1}^{k+l_c}}^{(k)}}_{T_{k+1}^{k+l_c}} \ket*{\Psi_{j_{k-l_c}^{k+l_c}}^{(k)}}_{T_k^{k+l_c}}}^2 \\
        			&\geq 1-\xi_{j_k}, \quad \forall j_k.
        		\end{aligned}
        	\end{equation}
        	
        	Then, the phase-error rate bound in~\cref{eq:general_eph_bound_correl} holds for this correlated scenario.
        \end{corollary}
        
        Note that \cref{cor:fidelity bound} places no structural assumption on the correlated source (beyond a finite correlation length $l_c$, which can be lifted by applying \cref{lem:trace_distance_bound} below). In particular, valid choices for the states $\ket*{\phi_{j_{k-l_c}^{k+l_c}}^{(k)}}_{T_k}$ and $\ket*{\lambda_{j_{k-l_c}^{k-1},j_{k+1}^{k+l_c}}^{(k)}}_{T_{k+1}^{k+l_c}}$ always exist (e.g., one may take $\ket*{\phi_{j_{k-l_c}^{k+l_c}}^{(k)}}_{T_k} = \ket*{\phi_{j_k}}_{T_k}$ and any fixed state for $\ket*{\lambda_{j_{k-l_c}^{k-1},j_{k+1}^{k+l_c}}^{(k)}}_{T_{k+1}^{k+l_c}}$), and once these are fixed, the fidelities in \cref{eq:correlations_general} are always well defined. Since the $\xi_{j_k}$ are free parameters of the uncorrelated security proof, \cref{eq:correlations_general} can then always be satisfied by simply setting $\xi_{j_k}$ to the worst-case fidelity over all rounds and setting histories. Consequently, \cref{cor:fidelity bound} can be applied to \emph{any} source with correlations of length $l_c$; the choice of $\ket*{\phi_{j_{k-l_c}^{k+l_c}}^{(k)}}_{T_k}$ and $\ket*{\lambda_{j_{k-l_c}^{k-1},j_{k+1}^{k+l_c}}^{(k)}}_{T_{k+1}^{k+l_c}}$ merely determines the achievable values of $\xi_{j_k}$, and hence the tightness of the resulting key-rate bound. This motivates optimizing this choice to maximize the fidelities in \cref{eq:correlations_general}; in the example below, we show how to do this for a correlations model of particular experimental interest.

        	\textit{Example: LTI correlations.}---As a concrete application of our framework, we consider correlations arising from modeling a BB84 phase modulator as a linear time-invariant (LTI) system~\cite{agulleiroModelingCharacterization2025}. Due to linearity, the encoding phase for the $k$-th pulse conditioned on the full setting history $j_1^k$ decomposes as
        	\begin{equation}
        		\label{eq:theta_j1k}
        		\theta_{j_1^k} = \hat\theta_{j_k} + \sum_{l=1}^{k-1} \delta^{(l)}_{j_{k-l}}.
        	\end{equation}
        	Here, $\hat\theta_{j_k}$ is the phase that would be encoded if the system had been in its baseline state prior to the encoding (which may still deviate from the ideal BB84 phase due to the imperfect response of the modulator) and $\delta^{(l)}_{j_{k-l}}$ represents the residual contribution from setting $j_{k-l}$ chosen $l$ rounds earlier. Assuming for simplicity no encoding side channels beyond correlations (though our framework can also incorporate them), the emitted state in round $k$ is
        	\begin{equation}
        		\label{eq:psi_j1k}
        		\ket*{\psi_{j_1^k}^{(k)}}_{T_k} = \cos\ (\theta_{j_1^k}) \ket{0}_{T_k} + \sin\ (\theta_{j_1^k}) \ket{1}_{T_k}.
        	\end{equation}
        	As shown in Ref.~\cite[Appendix~D]{agulleiroModelingCharacterization2025}, one can experimentally obtain an exponential bound on the correlation strength, i.e., 
        	\begin{equation}
        		\label{eq:delta_bound}
        		\abs{\delta^{(l)}_{j} - \delta^{(l)}_{j'}} \leq  \sqrt{\xi_l}, \quad \text{with} \quad
        		\xi_l = \xi_1 e^{-C (l-1)},
        	\end{equation}
        	where $j,j' \in \mathcal{J}$, $\xi_1$ is the nearest-neighbor correlation strength and $C > 0$ is a decay constant determined by the modulator's impulse response.
        	
        	As a first step, consider for now an idealized scenario where correlations vanish beyond some length $l_c$, i.e., $\xi_l = 0$ for $l > l_c$. The family of correlated states $\{\ket*{\psi_{j_{k-l_c}^k}^{(k)}}_{T_k}\}_{j_k}$ has a Gram matrix independent of both $k$ and the past settings $j_{k-l_c}^{k-1}$, matching that of the reference family $\{\ket{\phi_{j}}\}_{j}$ with
        	\begin{equation}
        		\label{eq:phi_j_def}
        		\ket{\phi_j} = \cos\  (\hat\theta_{j}) \ket{0} + \sin\ (\hat\theta_{j}) \ket{1}.
        	\end{equation}
        	To apply \cref{cor:fidelity bound}, we identify $\ket*{\phi_{j_{k-l_c}^{k+l_c}}^{(k)}}_{T_k} \mapsto \ket*{\psi^{(k)}_{j_{k-l_c}^{k}}}_{T_k}$ and
        	\begin{equation}
        		\label{eq:lambda_def}
        		\ket*{\lambda_{j_{k-l_c}^{k-1},j_{k+1}^{k+l_c}}^{(k)}}_{T_{k+1}^{k+l_c}} \mapsto \bigotimes_{m=k+1}^{k+l_c} \ket*{\psi^{(m)}_{j_{m-l_c}^{k-1},\, j^*,\, j_{k+1}^{m}}}_{T_m},		\end{equation}
        	where $j^* \in \mathcal{J}$ is an arbitrary fixed setting for the $k$-th pulse.  A straightforward calculation (see \cref{app:proof_xi}) then shows that \cref{eq:correlations_general} holds with $\xi_{j_k} = \xi$, $\forall j_k$, where
        	\begin{equation}
        		\label{eq:xi_jk_bound}
        		\xi \coloneqq \sum_{l=1}^{l_c} \xi_l = \frac{\xi_1 (1 - e^{-C l_c})}{1 - e^{-C}}.
        	\end{equation}
        	Thus, for finite-length correlations, one could directly apply \cref{cor:fidelity bound} to extend an uncorrelated proof that accommodates fidelity bounds of the form in \cref{eq:epsilon_side} to the reference states in \cref{eq:phi_j_def}, such as the proof in Ref.~\cite{curras-lorenzoSecurityFramework2025}.
        	
        	Of course, real sources may exhibit correlations of unbounded length, even if their strength may decay exponentially. To handle this, we introduce a freely chosen truncation length $l_c^\mathrm{eff}$, and account for the neglected long-range correlations beyond that length by slightly increasing the failure probability of the phase-error rate bound. Note that, unlike the previous approach of Ref.~\cite{pereiraQuantumKey2024}, where the truncation argument and trace-distance correction are applied as a separate step outside the phase-error proof, our framework incorporates this correction directly into the phase-error-rate bound, yielding a single self-contained statement that can be used to directly prove security.
        	\begin{lemma}
        		\label{lem:trace_distance_bound}
        		Let $\ket*{\Psi_N^{(\infty)}}_{A_1^N T_1^N}$ be the source-replacement state for a source with unbounded correlations, and let $\ket*{\Psi_N^{(l_c^\mathrm{eff})}}_{A_1^N T_1^N}$ be the corresponding state for a fictitious source with correlations truncated at length $l_c^\mathrm{eff}$. If the trace distance between these two states satisfies
        		\begin{equation}
        			T\Big(\ketbra*{\Psi_N^{(\infty)}},\ketbra*{\Psi_N^{(l_c^\mathrm{eff})}}\Big) \leq d,
        			\label{eq:def_d}
        		\end{equation}
        		and the phase-error bound in~\cref{eq:general_eph_bound} holds for the truncated source with failure probability $\epsilon$, then the same bound holds for the actual source with failure probability $\epsilon + d$.
        	\end{lemma}
        	To apply this result, we bound the trace distance $d$ between the actual and truncated source-replacement states. For the exponential model in~\cref{eq:delta_bound}, a straightforward calculation (see \cref{app:proof_d}) shows that
        	\begin{equation}
        		\label{eq:trace_distance_explicit}
        		d \leq \sqrt{N} \sum_{l = l_c^\mathrm{eff} + 1}^{\infty} \sqrt{\xi_l} = \frac{\sqrt{N \xi_1} e^{-C l_c^\mathrm{eff} / 2}}{1 - e^{-C/2}}.
        	\end{equation}
        	  Inverting this relation, to achieve a target $d$, one should choose the effective correlation length as
        	\begin{equation}
        		l_c^\mathrm{eff} = \bigg\lceil\frac{1}{C} \ln \left(\frac{N \xi_1}{d^2(1-e^{-C/2})^2}\right)\bigg\rceil.
        		\label{eq:eff_correlation_length}
        	\end{equation}
        	Note that, once $d$ is selected, $l_c^\mathrm{eff}$ can be deterministically calculated using \cref{eq:eff_correlation_length}, since $C$, $N$ and $\xi_1$ are known parameters. Then, applying \cref{cor:fidelity bound} to the fictitious source with correlations truncated at length $l_c^\mathrm{eff}$, and then applying \cref{lem:trace_distance_bound}, one obtains a secret key of length 
            	\begin{equation}
        		\begin{aligned}
        			\bm{l} &= \bm{n_K} \big[1-h\big( \mathcal{E}'_{\mathrm{ph}}  (\bm{\vec n};N,\epsilon, l_c^\mathrm{eff})\big)\big] - \lambda_\mathrm{EC}(\bm{\vec n}) \\
        			&- 2 \log_2\big(1/(2\varepsilon_\mathrm{PA})\big)-\log_2(2/\varepsilon_\mathrm{EV}),
        		\end{aligned}
        		\label{eq:sklength_unbounded}
        	\end{equation}
            with security parameter $(\varepsilon_\mathrm{corr}+\varepsilon_\mathrm{sec})$, where $\varepsilon_\mathrm{corr} = \varepsilon_\mathrm{EV}$ and $\varepsilon_\mathrm{sec} = 2 \sqrt{\epsilon+d} + \varepsilon_\mathrm{PA}$.
        

            We have focused on exponential decay because it arises naturally for bandwidth-limited optical modulators under the LTI model, which is supported by the experimental results in~\cite{agulleiroModelingCharacterization2025}. However, the approach above is not restricted to exponential decay: Lemma~\ref{lem:trace_distance_bound} applies to any correlation model, and the truncation is effective whenever the tail sum $\sum_{l > l_c^{\mathrm{eff}}} \sqrt{\xi_l}$ converges fast enough to make $d$ negligible for moderate values of $l_c^\mathrm{eff}$. For exponential decay, $l_c^{\mathrm{eff}} \sim \frac{1}{C} \ln N$, which grows only logarithmically with the block size $N$. For slower decay models, the scaling of $l_c^{\mathrm{eff}}$ with $N$ may be worse, which would reduce the usefulness of the truncation approach. While we are not aware of realistic QKD modulator models exhibiting such slowly decaying correlations, determining whether secure key distribution remains possible under such models is an interesting open question.
        
            Moreover, practical sources operating at very high clock rates may exhibit nonlinear pattern-dependent effects beyond the LTI description. Our framework is fully compatible with such nonlinear correlation models: the theoretical results (\cref{cor:main,cor:fidelity bound,lem:trace_distance_bound}) place no restrictions on the functional form of the history dependence, and apply whenever the multi-round correlated states can be experimentally characterized sufficiently to verify the relevant admissibility conditions. Once a validated model for such nonlinearities is available, one can directly apply our framework to obtain the corresponding key rates without any modification.

        	\textit{Key-rate simulations for BB84.}---We now demonstrate our framework by simulating the achievable secret-key rate for a BB84 protocol with a source suffering from LTI correlations, combining \cref{cor:fidelity bound,lem:trace_distance_bound} with the uncorrelated security proof in Ref.~\cite{curras-lorenzoSecurityFramework2025} \footnote{Note that, while the core security proof proposed by Ref.~\cite{curras-lorenzoSecurityFramework2025} assumes uncorrelated sources, this work then argues that the proof could be extended to correlated sources using the approach of Ref.~\cite{pereiraQuantumKey2020}. However, this extension suffers from the limitations identified in our introduction, and by applying our framework directly to the uncorrelated proof, we overcome these limitations.}. We consider a source that emits states of the form in \cref{eq:psi_j1k} with baseline phases $\hat \theta_{j} = (1+\delta_\mathrm{SPF}/\pi) \varphi_j$, where $\varphi_j \in \{0,\pi/4,\pi/2,3\pi/4\}$ is the ideal BB84 phase for setting $j \in \{0_Z, 0_X, 1_Z, 1_X\}$ and $\delta_\mathrm{SPF}$ parametrizes state-preparation flaws arising from the imperfect modulator response. The correlations $\delta^{(l)}_{j_{k-l}}$ are assumed to satisfy the exponential bound in \cref{eq:delta_bound}. We use the standard BB84 channel model in \cite{pereiraOptimalKey2025} with parameters: error-correction inefficiency $f = 1.16$, dark-count probability $p_d = 10^{-6}$, detector efficiency $\eta_d = 0.73$, and $N = 10^{12}$ transmitted signals. The state-preparation flaw is set to $\delta_\mathrm{SPF} = 0.068$ \cite{honjoDifferentialphaseshiftQuantum2004,xuExperimentalQuantum2015}, and for the correlations model, we consider $\xi_1 \in \{10^{-3}, 10^{-6}\}$ and $C = 12.7$~\cite{agulleiroModelingCharacterization2025}. The correctness and secrecy parameters of the final key are set to $\varepsilon_\mathrm{corr} = \varepsilon_\mathrm{sec} = 10^{-10}$. The results are shown in \cref{fig:results}. As expected, the secret-key rate drops as $\xi_1$ increases. Moreover, for the exponential decay model, the impact of the correlations is dominated by the first few correlation terms, since $\xi_l$ quickly becomes negligible with $l$. Thus, incorporating correlations of unbounded length results in almost no penalty compared to the finite-length case.
        	
        	\begin{figure}
        		\includegraphics[width=8.5cm]{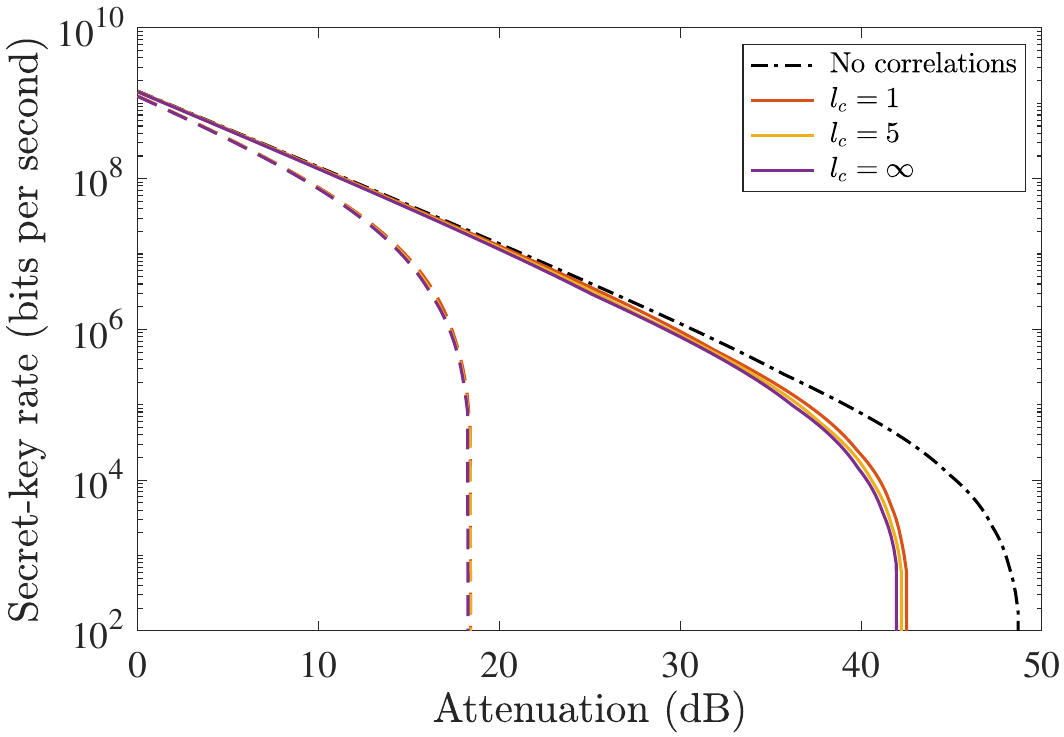} 
        		\caption{Secret-key rate versus channel attenuation for a BB84 protocol with a source suffering from LTI correlations, assuming the exponential decay model in \cref{eq:delta_bound}. The secret-key rates are obtained by combining \cref{cor:fidelity bound} with the uncorrelated security analysis of Ref.~\cite{curras-lorenzoSecurityFramework2025}. The solid lines correspond to $\xi_1 = 10^{-6}$ and the dashed lines correspond to $\xi_1 = 10^{-3}$. The lines labelled by $l_c=1$ and $l_c=5$ consider correlations truncated artificially at length $l_c$ (i.e.\ $\xi_l = 0$ for $l > l_c$), while $l_c = \infty$ corresponds to unbounded correlations handled via \cref{lem:trace_distance_bound}. The dashed dotted line corresponds to the ideal case of no correlations.}
        		\label{fig:results}
        	\end{figure}

        	\textit{Conclusion.}---We have established a rigorous mathematical framework for extending phase-error-rate bounds to scenarios with encoding correlations, which can then be directly used to prove the security of QKD protocols in their presence. Our framework resolves key limitations of previous approaches: it eliminates the need for multiple privacy amplification steps, circumvents contested composability arguments, handles unbounded correlation lengths naturally, and applies systematically to any existing phase-error-estimation-based analysis that considers appropriate admissibility conditions on the emitted states. We anticipate that our framework will prove valuable for the security analysis of practical QKD implementations where such correlations are unavoidable, and may extend to other scenarios beyond QKD involving temporally correlated sources.

        	\textit{Acknowledgements.}---We thank Devashish Tupkary, Shlok Nahar and Víctor Zapatero for valuable discussions, and Davide Rusca for insights on LTI correlations. This work was supported by the Galician Regional Government (consolidation of research units: atlanTTic), the Spanish Ministry of Economy and Competitiveness (MINECO), the Fondo Europeo de Desarrollo Regional (FEDER) through the grant No.~PID2024-162270OB-I00, MICIN with funding from the European Union NextGenerationEU (PRTRC17.I1) and the Galician Regional Government with own funding through the “Planes Complementarios de I+D+I con las Comunidades Autonomas” in Quantum Communication, the “Hub Nacional de Excelencia en Comunicaciones Cuanticas” funded by the Spanish Ministry for Digital Transformation and the Public Service and the European Union NextGenerationEU, the European Union’s Horizon Europe Framework Programme under the Marie Sklodowska-Curie Grant No.~101072637 (Project QSI), the project “Quantum Secure Networks Partnership” (QSNP, grant agreement No 101114043) and the European Union via the European Health and Digital Executive Agency (HADEA) under the Project QuTechSpace (grant 101135225) and the Project IberianQCI (grant 101249593), as well as the Programa de Cooperación Interreg VI-A España--Portugal (POCTEP) 2021--2027 through the project QUANTUM\_IBER\_IA. G.C.-L. acknowledges funding from the European Union's Horizon Europe research and innovation programme under the Marie Skłodowska-Curie Postdoctoral Fellowship grant agreement No.\ 101149523. K.T. acknowledges support from JSPS KAKENHI Grant Numbers 23K25793 and 23H01096.

            \appendix

        	\section{Application to B92}
            \label{app:b92}

            Consider a security proof for the B92 protocol, such as Ref.~\cite{tamakiUnconditionalSecurity2004}, that assumes Alice emits some characterized states $\ket{\phi_0}$ and $\ket{\phi_1}$, and suppose that these states satisfy
        	\begin{equation}
        		\label{eq:phi_0_phi_1_fid}
        		\abs{\braket{\phi_1}{\phi_0}}^2 = c.
        	\end{equation}
        	Now consider a practical implementation in which Alice has a flawed but uncorrelated source emitting states $\ket*{\psi_0^{(k)}}_{T_k}$ and $\ket*{\psi_1^{(k)}}_{T_k}$ in round $k$ that satisfy
        	\begin{equation}
        		\abs*{\braket*{\psi^{(k)}_1}{\psi^{(k)}_0}_{T_k}}^2 \geq c.
        	\end{equation}
        	Then there exists a completely-positive trace-preserving (CPTP) map $\mathcal{M}_{k}$ such that $\ketbra*{\psi_j^{(k)}}_{T_k} = \mathcal{M }_{k}(\ketbra*{\phi_j}_{T_k})$ for $j \in \{0,1\}$~\cite{jiangSidechannelSecurity2024}. Since this map can be absorbed into Eve's attack channel, the proof applies under the general admissibility set
        	\begin{equation}
        		\mathcal{S} = \big\{\{\ket{\psi_j}\}_{j\in \{0,1\}} : \abs{\braket{\psi_1}{\psi_0}}^2 \geq c\big\}.
        	\end{equation}
        	
        	To directly extend the proof to correlated sources via Corollary~\ref{cor:main}, the requirement is that the multi-round correlated states defined in \cref{eq:multi_round_correlated_states_def} satisfy
        	\begin{equation}
        		\label{eq:b92_correlated_condition}
        		\abs{\braket*{\Psi_{j_{k-l_c}^{k-1}, 1, j_{k+1}^{k+l_c}}}{\Psi_{j_{k-l_c}^{k-1}, 0, j_{k+1}^{k+l_c}}}_{T_k^{k+l_c}}}^2 \geq c,
        	\end{equation}
        	for all fixed past settings $j_{k-l_c}^{k-1}$ and future settings $j_{k+1}^{k+l_c}$. This condition ensures that the correlated states $\{\ket*{\Psi_{j_{k-l_c}^{k+l_c}}}_{T_k^{k+l_c}}\}_{j_k \in \{0,1\}}$ are isometrically equivalent to an acceptable family $\{\ket{\psi_0}, \ket{\psi_1}\} \in \mathcal{S}$, since both families share the same inner product structure by construction.
        	
        	\section{Application to side-channel-secure (SCS) QKD}
            \label{app:scs}

            In SCS-QKD protocols \cite{wangPracticalLongDistance2019,jiangSidechannelSecurity2024}, since Alice and Bob send only two states each, one can also apply the same argument to trivially extend an uncorrelated security proof to the acceptability set in \cref{eq:b92_correlated_condition} for the states emitted by each user. Our results are not directly applicable to the existing security proofs in \cite{wangPracticalLongDistance2019,jiangSidechannelSecurity2024}, however, since these are not based on obtaining a phase-error bound of the form in \cref{eq:general_eph_bound} that holds for general eavesdropping attacks, but are based on the application of the Postselection Technique to extend the proof from collective to general attacks. Still, a recently proposed variant of the protocol \cite{shanPracticalPhaseCoding2023} performs better in some contexts and its security proof is based on obtaining a phase-error rate bound as in \cref{eq:general_eph_bound}, and thus the security proof in \cite{shanPracticalPhaseCoding2023} can be extended directly to correlated sources via \cref{cor:main}, with the condition that the multi-round correlated states emitted by Alice and Bob satisfy a bound as in \cref{eq:b92_correlated_condition}.


        	\section{Proof of \cref{eq:xi_jk_bound}}
            \label{app:proof_xi}

            After identifying $\ket*{\phi_{j_{k-l_c}^{k+l_c}}^{(k)}}_{T_k} \mapsto \ket*{\psi^{(k)}_{j_{k-l_c}^{k}}}_{T_k}$ and \cref{eq:lambda_def}, the inner product squared in \cref{eq:correlations_general} becomes
        	\begin{align}  
        		&\abs{\bra*{\psi^{(k)}_{j_{k-l_c}^{k}}}_{T_k} \bigotimes_{m=k+1}^{k+l_c} \bra*{{\psi^{(m)}_{j_{m-l_c}^{k-1},\, j^*,\, j_{k+1}^{m}}}}_{T_m} \ket*{\Psi_{j_{k-l_c}^{k+l_c}}^{(k)}}_{T_k^{k+l_c}}}^2 \nonumber \\
        		&=\prod_{m=k+1}^{k+l_c} \abs{\braket*{\psi^{(m)}_{j_{m-l_c}^{k-1},\, j^*,\, j_{k+1}^{m}}}{\psi^{(m)}_{j_{m-l_c}^{k-1},\, j_k,\, j_{k+1}^{m}}}_{T_m}}^2 \nonumber \\
        		&= \prod_{l=1}^{l_c} \cos^2 \!\big(\delta^{(l)}_{j_k} - \delta^{(l)}_{j^*}\big) = \prod_{l=1}^{l_c} \Big(1-\sin^2\!\big(\delta^{(l)}_{j_k} - \delta^{(l)}_{j^*}\big)\Big) \nonumber \\
        		&\ge \prod_{l=1}^{l_c} \Big(1-\big(\delta^{(l)}_{j_k} - \delta^{(l)}_{j^*}\big)^2\Big) \ge \prod_{l=1}^{l_c} (1-\xi_l) \nonumber \\
        		&\ge 1-\sum_{l=1}^{l_c}\xi_l ~=: 1-\xi.
        		\label{eq:proof_eq15}
        	\end{align}
        	Then, by substituting the exponential model in \cref{eq:delta_bound} into \cref{eq:proof_eq15}, we obtain \cref{eq:xi_jk_bound}.

        	\section{Proof of \cref{eq:trace_distance_explicit}}
            \label{app:proof_d}

            Let
        	\begin{equation}
        		\label{eq:source_replacement_correlated_lc}
        		\ket*{\Psi_N^{(l_c^\mathrm{eff})}}_{A_1^N T_1^N} = \sum_{j_1^N} \bigotimes_k \sqrt{p_{j_k}} \ket*{j_k}_{A_k} \ket*{\psi^{(k)}_{j_{k-l_c^\mathrm{eff}}^k}}_{T_k},
        	\end{equation}
        	where we have defined
        	\begin{equation}
        		\ket*{\psi^{(k)}_{j_{k-l_c^\mathrm{eff}}^k}}_{T_k} = \ket*{\psi^{(k)}_{j_{1}^k}}_{T_k} \bigg\vert_{j_1 = j_2 = ... = j_{k-l_c^\mathrm{eff}-1} = j^{*}},
        	\end{equation}
        	with $j^* \in \mathcal{J}$ being an arbitrary fixed setting. Using similar derivations as in \cite[Appendix B]{pereiraQuantumKey2024}, we have that
        	\begin{equation}
        		\begin{gathered}
        			\abs{\braket*{\psi^{(k)}_{j_{k-l_c^\mathrm{eff}}^k}}{\psi^{(k)}_{j_{1}^k}}_{T_k}}^2 \geq 1- \Big(\sum_{l = l_c^\mathrm{eff}+1}^{\infty} \sqrt{\xi_l}\Big)^2.
        		\end{gathered}
        	\end{equation}
        	Therefore,
        	\begin{equation}
        		\begin{gathered}
        			\abs{\braket*{\Psi_N^{(l_c^\mathrm{eff})}}{\Psi_N^{(\infty)}}_{A_1^N T_1^N}}^2 \geq \bigg[1- \Big(\sum_{l = l_c^\mathrm{eff}+1}^{\infty} \sqrt{\xi_l}\Big)^2\bigg]^N.
        		\end{gathered}
        	\end{equation}
        	Thus, using the fact that for $x \in [0,1]$, $1-(1-x)^N \leq N x$,
        	\begin{equation}
        		d = \sqrt{1-\abs{\braket*{\Psi_N^{(l_c^\mathrm{eff})}}{\Psi_N^{(\infty)}}_{A_1^N T_1^N}}^2} \leq \sqrt{N} \sum_{l = l_c^\mathrm{eff}+1}^{\infty} \sqrt{\xi_l}.
        		\label{eq:proof_eq19}
        	\end{equation}
        	Then, by substituting the exponential model in \cref{eq:delta_bound} into \cref{eq:proof_eq19} and evaluating the resulting convergent series, we obtain \cref{eq:trace_distance_explicit}.
        	
        	\bibliography{refs,refs_extra}

	 \clearpage

     \onecolumngrid
	 \ifarXiv
	 

\section*{Supplementary material}

## Appendix A: Security framework based on phase-error estimation

Here, we present the security framework based on phase-error estimation and show that it remains valid even when the emitted pulses are correlated. Consider a general prepare-and-measure QKD protocol where, in each round $k \in \{1,...N\}$, Alice selects a setting choice $j_k \in \mathcal{J}$ with probability $p_{j_k}$, and sends a quantum state to Bob through an insecure quantum channel. Due to potential correlations in Alice's source (arising from, e.g., memory effects in the modulator), the state emitted in round $k$ may depend not only on the current setting choice $j_k$, but also on the previous history $j_1^{k-1}$. We denote this state by $|\psi_{j_1^k}^{(k)}\rangle_{T_k}$, where $j_1^k := j_1...j_k$.

As for Bob, we consider that, in each round, he chooses $\beta_k \in \{\text{key}, \text{test}\}$ with probability $p_{\beta_k}$, and performs a POVM $\vec{\Gamma}_{\beta_k}$. Here, $\vec{\Gamma}_{\text{test}}$ is a POVM with any number of outcomes[3] (possibly including a non-detection outcome, indicating that the data from that round will not be used), while $\vec{\Gamma}_{\text{key}} := \{\Gamma_0^{\text{key}}, \Gamma_1^{\text{key}}, \Gamma_\perp^{\text{key}}\}$ is a POVM with three elements, corresponding respectively to bit 0, bit 1, and non-detection.

For concreteness, we consider that Alice and Bob extract their sifted keys from the rounds in which $j_k \in \{0,1\}$, $\beta_k = \text{key}$, and Bob obtains a detection (i.e., outcome 0 or 1 rather than $\perp$)[4]. This allows us to define the following protocol:

> **Actual Protocol (prepare-and-measure)**
>
> 1. *State preparation:* For each round $k \in \{1,...,N\}$, Alice randomly selects a setting $j_k \in \mathcal{J}$ with probability $p_{j_k}$, prepares a quantum state $|\psi_{j_1^k}^{(k)}\rangle_{T_k}$, and sends it to Bob through the quantum channel.
>
> 2. *Eve's attack:* Eve performs the most general attack allowed by quantum mechanics on the transmitted systems $T_1^N$, and re-sends some output systems $B_1^N$ to Bob, while keeping an ancillary system $E$.
>
> 3. *Measurement and basis choice:* For each round $k$, Bob chooses $\beta_k \in \{\text{key}, \text{test}\}$ with probability $p_{\beta_k}$ and performs the corresponding POVM: $\vec{\Gamma}_{\text{key}}$ if $\beta_k = \text{key}$, or $\vec{\Gamma}_{\text{test}}$ if $\beta_k = \text{test}$. Bob records his measurement outcome.
>
> 4. *Sifting:* Bob announces which rounds resulted in a detection, along with his choice of $\beta_k$ for each detected round. For the detected rounds with $\beta_k = \text{test}$, and for the detected rounds with $\beta_k = \text{key}$ and $j_k \notin \{0,1\}$, Alice announces her setting $j_k$. For the detected rounds with $\beta_k = \text{key}$ and $j_k \in \{0,1\}$, Alice announces $\alpha_k = \text{key}$.
>
> 5. *Bit-error-rate estimation:* Alice and Bob choose a random subset $\mathcal{D}_{\text{key}}$ of the detected rounds in which $\alpha_k = \beta_k = \text{key}$, which will be used to generate the sifted key, and announce which rounds belong to this subset. For the remaining rounds with $\alpha_k = \beta_k = \text{key}$, Alice and Bob announce their bit values (Alice announces $j_k \in \{0,1\}$ and Bob announces his measurement outcome $b_k \in \{0,1\}$) to estimate the bit-error rate.
>
> 6. *Sifted key formation:* For each round $k \in \mathcal{D}_{\text{key}}$, Alice's sifted key bit is $j_k$ and Bob's sifted key bit is his measurement outcome $b_k$.
>
> 7. *Variable-length decision:* Let $\vec{n}$ denote all the data announced by Alice and Bob until this point. Using this data, Alice and Bob compute $\lambda_{\text{EC}}(\vec{n})$ (the number of bits to be revealed in one-way error correction) and $l(\vec{n})$ (the length of the final key to be produced, see Eq. (A3)). Aborting corresponds to $l(\vec{n}) = 0$.
>
> 8. *Error correction and error verification:* Alice and Bob implement a one-way error correction protocol that reveals $\lambda_{\text{EC}}(\vec{n})$ bits of information. They implement error verification by using a common and randomly

---

[3] Note that, in practice, Bob may perform more than one test POVM, but the act of choosing between various test POVMs and then performing one of them can be mathematically described by a single POVM, so there is no loss of generality in assuming one single test POVM.

[4] More generally, one could consider situations in which the decision of whether or not a round is used for sifted key extraction is more complicated (e.g., by using additional auxiliary random variables) and/or situations in which the algorithm to extract the sifted key bits themselves is more complicated; see e.g. [38] for a more general and abstract description of a QKD protocol and its associated phase-error estimation protocol. To cover such situations, one simply needs to modify Steps 1–5 in the *Actual protocol (source-replaced)* (defined later) appropriately so that they determine a set of rounds $\mathcal{D}_{\text{key}}$ and a shared state between Alice and Bob such that, when they perform their sifted-key measurements in Step 6, the statistics are equivalent to those of the *Actual protocol (prepare-and-measure)*. We remark that all our results in Appendix B for extending a security proof from the uncorrelated to the correlated scenario still apply after such modifications, since the proof of these results does not depend on Alice's and Bob's specific actions in Steps 1–5, as long as the equivalence to the actual protocol is maintained.

   selected hash function from a two-universal family of output length $\log_2(2/\varepsilon_{\text{EV}})$ bits, having one of the parties announce the result, and comparing their values.

9. *Privacy amplification:* If error verification succeeds, Alice and Bob select a random hash function from a two-universal family and apply it to their sifted key to obtain a final key of length $l(\vec{n})$.

To analyze the security of the above protocol, we employ the source-replacement technique. The key observation is that Alice's prepare-and-measure procedure in the actual protocol is equivalent to the following: Alice first prepares the global entangled state

$$|\Psi_N\rangle_{A_1^N T_1^N} = \sum_{j_1^N} \bigotimes_k \sqrt{p_{j_k}} |j_k\rangle_{A_k} |\psi_{j_1^k}^{(k)}\rangle_{T_k}, \tag{A1}$$

(which is Eq. (1) in the main text), sends the photonic systems $T_1^N := T_1...T_N$ through the quantum channel, and then measures each system $A_k$ in the $\{|j_k\rangle_{A_k}\}_{j_k \in \mathcal{J}}$ basis to determine her setting choice. Since this measurement commutes with all operations on the $T_k$ systems (including Eve's attack and Bob's measurements), the statistical outcomes of the protocol are identical whether Alice measures before or after transmission.

Furthermore, for the security analysis, it is convenient to decompose Bob's key POVM $\vec{\Gamma}_{\text{key}} = \{\Gamma_0^{\text{key}}, \Gamma_1^{\text{key}}, \Gamma_\perp^{\text{key}}\}$ into two steps: first, a filter operation $\{F, \mathbb{I} - F\}$ with $F = \Gamma_0^{\text{key}} + \Gamma_1^{\text{key}}$ that determines whether a detection occurs, followed by a two-outcome POVM $\{G_0, G_1\}$ that determines the bit value conditional on detection[5]. This decomposition does not change the measurement statistics, but it allows us to defer Bob's bit-value measurement to a later stage in the protocol, which is useful for relating the actual protocol to the phase-error estimation protocol.

Using these observations, we can define the following source-replaced version of the protocol, which is statistically equivalent to the **Actual Protocol (prepare-and-measure)**:

**Actual Protocol (source replaced)**

1. *State preparation:* Alice prepares her global entangled state $|\Psi_N\rangle_{A_1^N T_1^N}$ and sends the photonic systems $T_1^N$ through the quantum channel.

2. *Eve's attack:* Eve performs the most general attack allowed by quantum mechanics on the transmitted systems $T_1^N$, and re-sends some output systems $B_1^N$ to Bob, while keeping an ancillary system $E$.

3. *Detection and test measurements:* For each round, Bob decides $\beta_k \in \{\text{key}, \text{test}\}$ with probability $p_{\beta_k}$. If $\beta_k = \text{key}$, he applies the filter $\{F, \mathbb{I} - F\}$, and if $\beta_k = \text{test}$, he measures $\vec{\Gamma}_{\text{test}}$. Based on these outcomes, Bob announces which rounds are detected, and his choice of $\beta_k$ in the detected rounds.

4. *Key/test determining and setting announcement:* For each detected round with $\beta_k = \text{test}$, Alice measures her system $A_k$ in the $\{|j_k\rangle_{A_k}\}_{j_k \in \mathcal{J}}$ basis, and announces her outcome $j_k$. For each detected round with $\beta_k = \text{key}$, Alice attempts a projection onto the subspace spanned by $\{|0\rangle_{A_k}, |1\rangle_{A_k}\}$. If successful, Alice announces $\alpha_k = \text{key}$; otherwise, Alice measures her system $A_k$ in the $\{|j_k\rangle_{A_k}\}_{j_k \in \mathcal{J}}$ basis, and announces her outcome $j_k$.

5. *Bit-error-rate estimation:* Alice and Bob choose a random subset $\mathcal{D}_{\text{key}}$ of the detected rounds in which $\alpha_k = \beta_k = \text{key}$, which will be used to generate the sifted key, and announce this information. For the remaining rounds with $\alpha_k = \beta_k = \text{key}$, Alice measures $A_k$ in $\{|0\rangle_{A_k}, |1\rangle_{A_k}\}$ and Bob measures $\{G_0, G_1\}$, and both announce their results.

6. *Sifted-key measurements:* For each round $k \in \mathcal{D}_{\text{key}}$, Alice measures $A_k$ in $\{|0\rangle_{A_k}, |1\rangle_{A_k}\}$ and Bob measures $\{G_0, G_1\}$, and they define their respective sifted keys as their respective bit outcomes in these rounds.

7-9. Same as in **Actual protocol (prepare-and-measure)**.

To prove the security of the final key pair, we consider the following phase-error estimation protocol:

---

[5] These POVM elements are defined as $G_b := \sqrt{F^+} \Gamma_b^{\text{key}} \sqrt{F^+} + P_b$ for $b \in \{0,1\}$, where $F^+$ denotes the pseudoinverse of $F$, and $P_b$ are any two positive operators satisfying $\sum_{b \in \{0,1\}} P_b = \mathbb{I} - \Pi_F$, with $\Pi_F$ denoting the projector onto the support of $F$. See, e.g., [7] or [11, Appendix A].

> **Phase-error estimation protocol**
>
> 1-5. Same as in **Actual protocol (source replaced)**
>
> 6. *Phase-error measurements:* For each round $k \in \mathcal{D}_{\text{key}}$, Alice measures $A_k$ in $\{|+\rangle_{A_k}, |-\rangle_{A_k}\}$, where $|\pm\rangle_{A_k} = (|0\rangle_{A_k} \pm |1\rangle_{A_k})/\sqrt{2}$, and Bob measures $\{G_+, G_-\}$[6]. We denote the phase-error rate $e_{\text{ph}}$ as the fraction of events in which their outcomes differ.

The objective of a security proof based on phase-error estimation is to find a bound of the form

$$\Pr[e_{\text{ph}} > \mathcal{E}_{\text{ph}}(\vec{n}; N, \epsilon)] \leq \epsilon. \tag{A2}$$

Here, $e_{\text{ph}}$ is the random variable associated to the phase-error rate, $\vec{n}$ is the random vector representing the announced data in Steps 1-5, $N$ is the total number of transmitted rounds, $\epsilon$ is the failure probability of the bound, and $\mathcal{E}_{\text{ph}}$ is a function relating all these quantities. It is well known that a bound of the form in Eq. (A2) is enough to establish security using either entropic uncertainty relations (EUR) and the leftover hashing lemma (LHL) [18–20] or phase-error correction (PEC) arguments [21], even when the final key is allowed to be of variable length [7, 26–28]. Here, we consider the EUR+LHL framework, and in particular, we use the following result:

**Theorem 1** (Variable-length security of QKD protocols from EUR+LHL). *Suppose that, for a given source-replacement state $|\Psi_N\rangle_{A_1^N T_1^N}$, we have the guarantee that, in the Phase-error estimation protocol, Eq. (A2) holds for any eavesdropping attack. Let $\lambda_{\text{EC}}(\vec{n})$ be a function that determines the number of bits revealed in error correction, and let*

$$l(\vec{n}) = \max\left[0, n_K\left[1 - h(\mathcal{E}_{\text{ph}}(\vec{n}; N, \epsilon))\right] - \lambda_{\text{EC}}(\vec{n}) - 2\log_2\left(\frac{1}{2\varepsilon_{\text{PA}}}\right) - \log_2\left(\frac{2}{\varepsilon_{\text{EV}}}\right)\right], \tag{A3}$$

*be a function that determines the length of the final key, where $n_K$ is determined by $\vec{n}$, $h(x)$ is the binary entropy function $-x\log_2(x) - (1-x)\log_2(1-x)$ for $x \leq 1/2$ and $h(x) = 1$ otherwise. Then, if Alice and Bob run the Actual protocol (prepare-and-measure) using this choice of $\lambda_{\text{EC}}(\vec{n})$ and $l(\vec{n})$, the output key is $(2\sqrt{\epsilon} + \varepsilon_{\text{PA}} + \varepsilon_{\text{EV}})$-secure.*

*Proof.* The security of the **Actual protocol (prepare-and-measure)** follows from the security of the **Actual protocol (source replaced)**, since these two protocols are statistically equivalent.

Let $W$ be the classical register containing the outcome of the announced data vector $\vec{n}$, let $\Omega(\vec{n})$ be the event in which $\vec{n} = \vec{n}$ is observed, let $\rho_{|\Omega(\vec{n})}$ be the state shared by Alice, Bob and Eve before Step 7 in the *Actual protocol (source replaced)* conditional on $\Omega(\vec{n})$, let $\rho^{\text{virt}}_{|\Omega(\vec{n})}$ be the state shared by Alice, Bob and Eve at the end of the *Phase-error estimation protocol* conditional on $\Omega(\vec{n})$, and let

$$\kappa(\vec{n}) := \Pr[e_{\text{ph}} > \mathcal{E}_{\text{ph}}(\vec{n}; N, \epsilon) \mid \Omega(\vec{n})]. \tag{A4}$$

Also, let $Z_A^{n_K}$ be the register in $\rho_{|\Omega(\vec{n})}$ containing Alice's sifted key in the *Actual protocol (source replaced)*, and let $X_A^{n_K}$ ($X_B^{n_K}$) be the register in $\rho^{\text{virt}}_{|\Omega(\vec{n})}$ containing Alice's (Bob's) bit outcomes for the rounds in $\mathcal{D}_{\text{key}}$ in the *Phase-error estimation protocol*. Applying the EUR on the states conditional on $\Omega(\vec{n})$ and with smoothing parameter $\kappa(\vec{n})$ as in [7, Theorem 1] (see also [26, Supp. Note A]), we obtain

$$H_{\min}^{\sqrt{\kappa(\vec{n})}}(Z_A^{n_K} \mid WE)_{\rho_{|\Omega(\vec{n})}} \geq n_K - H_{\max}^{\sqrt{\kappa(\vec{n})}}(X_A^{n_K} \mid X_B^{n_K})_{\rho^{\text{virt}}_{|\Omega(\vec{n})}} \geq n_K\left[1 - h(\mathcal{E}_{\text{ph}}(\vec{n}; N, \epsilon))\right], \tag{A5}$$

where $E$ is the system containing Eve's side information (see Step 2 in the Actual Protocol (source replaced)). Then, security follows directly from [7, Theorem 4] after identifying $i \mapsto \vec{n}$, $j \mapsto \varnothing$, $\beta_i \mapsto n_K\left[1 - h(\mathcal{E}_{\text{ph}}(\vec{n}; N, \epsilon))\right]$, $\kappa_{(i,j)} \mapsto \kappa(\vec{n})$, $\varepsilon_{\text{AT}}^2 \mapsto \epsilon$, $\vec{Z} \mapsto Z_A^{n_K}$, $\vec{C} \mapsto W$ and $\vec{E} \mapsto E$. □

**Remark 1.** *Note that Theorem 1 applies regardless of the form of Alice's source-replacement state $|\Psi_N\rangle_{A_1^N T_1^N}$, and in particular, it holds even if Alice's source is correlated across rounds (see Eq. (A1)). To our knowledge, this is the first time that this is explicitly established.*

---

[6] Here, $\{G_+, G_-\}$ is an arbitrary two-outcome POVM, and does not necessarily correspond to anything that Bob does in the actual protocol. For the special case in which Bob performs two POVMs satisfying the basis-independent detection efficiency condition, i.e., when $\vec{\Gamma}^{\text{test}} := \{\Gamma_+^{\text{test}}, \Gamma_-^{\text{test}}, \Gamma_\perp^{\text{test}}\}$ and $\Gamma_\perp^{\text{key}} = \Gamma_\perp^{\text{test}}$, it is useful to define $G_b := \sqrt{F^+}\, \Gamma_b^{\text{test}}\, \sqrt{F^+} + P_b$ for $b \in \{+, -\}$, where $F^+$ denotes the pseudoinverse of $F$ and $P_b$ are any two positive operators satisfying $\sum_{b \in \{+, -\}} P_b = \mathbb{I} - \Pi_F$, with $\Pi_F$ denoting the projector onto the support of $F$. By doing so, Bob's fictitious measurement for the key rounds in the phase-error estimation protocol corresponds directly to his actual measurement for the test rounds, which simplifies significantly the security proof by allowing a direct application of random sampling arguments. However, in more general scenarios, including BB84-type protocols for which the basis-independent detection efficiency condition is not satisfied [2, 7, 11, 39, 40], one needs to define $\{G_+, G_-\}$ in a way in which it does not correspond exactly to a measurement that Bob does in the actual protocol, see [2, 7, 11, 39, 40] for more information.

**Appendix B: Extending existing phase-error bounds to incorporate encoding correlations**

In this section, we present our general framework to extend a phase-error-estimation-based security proof for imperfect but uncorrelated sources to handle encoding correlations as well. Our approach builds upon the round-partitioning strategy introduced in previous works [5, 6, 22, 23], but significantly generalizes it, puts it into a more rigorous theoretical foundation, and resolves many of its fundamental limitations highlighted in the Introduction of the main text.

The core insight underlying our framework is that correlations of length $l_c$ create dependencies only between rounds whose indices differ by at most $l_c$. By partitioning the protocol rounds into $(l_c+1)$ groups according to $I_w = \{k : k \equiv w \pmod{l_c+1}\}$, we ensure that rounds within each group are separated by at least $(l_c+1)$ positions, effectively eliminating direct correlations between them. This allows us to apply the uncorrelated analysis separately to each partition $I_w$ to establish an upper bound on its phase-error rate.

Then, we rigorously show that these individual phase-error-rate upper bounds can be combined to establish an upper bound on the *overall* phase-error rate. This is the key difference with respect to previous works, which consider a separate security proof and privacy amplification procedure for each partition, and then argue that the separate security proofs could be combined through composability arguments. The latter approach introduces both practical complications and theoretical concerns, as discussed in the Introduction of the main text, which our approach overcomes.

We present our results in a hierarchical structure, progressing from the most general formulation to increasingly specific and practical cases:

1. **Theorem 2** establishes the most general result, applicable to a general prepare-and-measure protocol where the uncorrelated security proof imposes general admissibility conditions on the global quantum state emitted for each possible sequence of setting choices. This theorem forms the mathematical foundation of our framework.

2. **Corollary 1** specializes the general theorem to the natural case in which the uncorrelated proof imposes admissibility conditions on the states emitted in individual rounds, rather than on the global state. This formulation matches the structure of most existing security proofs considering partially-characterized encoding imperfections and side channels. This is the result highlighted in the main text, as we consider it to be the most useful result in practice.

3. **Corollary 2** further specializes it to the practically important case where the admissibility condition is a fidelity bound between the actually emitted states and some reference states. This is considered in [6], and equivalent conditions are considered in [5, 23, 29, 30].

4. **Lemma 1** extends the framework to obtain a phase-error-rate upper bound even with unbounded correlation lengths.

For concreteness, our results focus on prepare-and-measure protocols. However, they extend naturally to interference-based protocols (also known as MDI-type protocols), see Remark 3 at the end of the Appendix.

**Theorem 2.** *Consider a prepare-and-measure QKD protocol with an uncorrelated source where Alice emits a global state*

$$|\Psi_{j_1^N}\rangle_{T_1^N} = \bigotimes_{k=1}^{N} |\psi_{j_k}^{(k)}\rangle_{T_k}, \tag{B1}$$

*when choosing a full sequence of setting choices $j_1^N \in \mathcal{J}^N$. Suppose that, for each $N$, there exists an admissibility set $\mathcal{S}_N$ of state families indexed by $j_1^N$ such that the phase-error rate bound*

$$\Pr[\boldsymbol{e_{\mathrm{ph}}} > \mathcal{E}_{\mathrm{ph}}(\vec{\boldsymbol{n}}; N, \epsilon)] \leq \epsilon, \tag{B2}$$

*holds for any eavesdropping attack as long as*

$$\{|\Psi_{j_1^N}\rangle_{T_1^N}\}_{j_1^N \in \mathcal{J}^N} \in \mathcal{S}_N, \tag{B3}$$

*where $\vec{\boldsymbol{n}}$ represents the random vector containing all the announced data.*

Now consider the analogous protocol with a source exhibiting correlations up to length $l_c$, where Alice emits a global state

$$|\Psi'_{j_1^N}\rangle_{T_1^N} = \bigotimes_{k=1}^{N} |\psi^{(k)}_{j_{k-l_c}^k}\rangle_{T_k}, \tag{B4}$$

conditional on choosing a sequence of setting choices $j_1^N \in \mathcal{J}^N$.

Partition the rounds $\{1, ..., N\}$ into $(l_c + 1)$ sets according to $I_w = \{k : k \equiv w \mod l_c + 1\}$, and define also the complementary sets $I_{\bar{w}} = \{k : k \not\equiv w \mod l_c + 1\}$. Then, for each $w$, define the subsequences $j_{I_w}$ and $j_{I_{\bar{w}}}$ of $j_1^N$ indexed by $I_w$ and $I_{\bar{w}}$, respectively. Suppose that, for every $w$ and every choice of $j_{I_{\bar{w}}}$, there exists an isometry $V^{j_{I_{\bar{w}}}}_{T_{I_w} \to T_1^N} : \mathcal{H}_{T_{I_w}} \to \mathcal{H}_{T_1^N}$ such that

$$\left\{ \left(V^{j_{I_{\bar{w}}}}_{T_{I_w} \to T_1^N}\right)^{\dagger} |\Psi'_{j_1^N}\rangle_{T_1^N} \right\}_{j_{I_w} \in \mathcal{J}^{N^{(w)}}} \in \mathcal{S}_{N^{(w)}}, \tag{B5}$$

where $N^{(w)} = |I_w|$ and $\mathcal{S}_{N^{(w)}}$ is the admissibility set defined by the uncorrelated proof for a protocol with $N^{(w)}$ transmitted rounds.

Then, the phase-error rate in the correlated scenario satisfies

$$\Pr\left[ e_{\mathrm{ph}} > \frac{\sum_{w=0}^{l_c} \mathbf{n}_K^{(w)} \mathcal{E}_{\mathrm{ph}}(\vec{\mathbf{n}}^{(w)}; N^{(w)}, \epsilon)}{\mathbf{n}_K} \right] \leq (l_c + 1)\epsilon, \tag{B6}$$

where $\vec{\mathbf{n}}^{(w)}$ is the restriction of the announced data vector $\vec{\mathbf{n}}$ to rounds in $I_w$, $\mathbf{n}_K^{(w)}$ is the number of sifted key bits from rounds in $I_w$, and $\mathbf{n}_K = \sum_{w=0}^{l_c} \mathbf{n}_K^{(w)}$ is the total number of bits in the sifted key.

*Proof.* Let us first review the phase-error estimation protocol for the uncorrelated scenario. The global source-replacement state generated by Alice can be written as

$$|\Psi_N\rangle_{A_1^N T_1^N} = \sum_{j_1^N} \sqrt{\Pr[j_1^N]} |j_1^N\rangle_{A_1^N} |\Psi_{j_1^N}\rangle_{T_1^N}, \tag{B7}$$

where $\Pr[j_1^N] = \prod_{k=1}^{N} p_{j_k}$. Then, Eve applies her global isometry $V_{T_1^N \to B_1^N E}$ and sends systems $B_1^N$ to Bob. We are interested in the state shared by Alice and Bob after Eve's attack, and thus we define a completely-positive trace-preserving (CPTP) map $\Phi_{T_1^N \to B_1^N}$ that consists of first applying $V_{T_1^N \to B_1^N E}$ and then tracing out Eve's ancillary system $E$. Thanks to this, we can write the state shared by Alice and Bob after Eve's attack as

$$\rho_{A_1^N B_1^N} = \Phi_{T_1^N \to B_1^N}\left( |\Psi_N\rangle\langle\Psi_N|_{A_1^N T_1^N} \right). \tag{B8}$$

Next, Alice and Bob perform measurements on their systems $A_1^N B_1^N$, through which they will learn the values of $e_{\mathrm{ph}}$ and $\vec{\mathbf{n}}$. We can define a simple two-outcome POVM $\{M^{\leq,\epsilon}_{A_1^N B_1^N}, M^{>,\epsilon}_{A_1^N B_1^N}\}$ that only checks whether $e_{\mathrm{ph}} \leq \mathcal{E}_{\mathrm{ph}}(\vec{\mathbf{n}}; N, \epsilon)$ or $e_{\mathrm{ph}} > \mathcal{E}_{\mathrm{ph}}(\vec{\mathbf{n}}; N, \epsilon)$. Using this, we can restate the phase-error rate bound in Eq. (B2) as the following guarantee: if Alice generates the global state in Eq. (B7) and Eq. (B3) holds, then, for any CPTP map $\Phi_{T_1^N \to B_1^N}$,

$$\mathrm{Tr}\left[ M^{>,\epsilon}_{A_1^N B_1^N} \Phi_{T_1^N \to B_1^N}\left( |\Psi_N\rangle\langle\Psi_N|_{A_1^N T_1^N} \right) \right] \leq \epsilon. \tag{B9}$$

Now, let's consider the analogous scenario with a correlated source. The global source-replacement state generated by Alice is now

$$|\Psi'_N\rangle_{A_1^N T_1^N} = \sum_{j_1^N} \sqrt{\Pr[j_1^N]} |j_1^N\rangle_{A_1^N} |\Psi'_{j_1^N}\rangle_{T_1^N}. \tag{B10}$$

Our strategy to upper bound the phase-error rate in the correlated scenario is to partition the protocol rounds into $(l_c + 1)$ sets $I_w$ indexed by $w \in \{0, 1, \ldots, l_c\}$, and upper-bounding the phase-error rate of each partition $I_w$ independently. To achieve this, we will show that each partition satisfies the conditions of the uncorrelated scenario when conditioning on any value of the settings $j_{I_{\bar{w}}}$ of the complementary set $I_{\bar{w}}$. As a tool to upper bound the phase-error rate of each partition $I_w$, we introduce a modified protocol in which Alice and Bob perform the phase-error measurements for the sifted key rounds in $I_w$, but perform the actual bit measurements for the rounds in $I_{\bar{w}}$.

> **$w$-th phase-error estimation protocol (PEEP)** (defined for each $w \in \{0, 1, \ldots, l_c\}$)
>
> 1-5. Same as in **Actual protocol (source replaced)**
>
> 6. *Measurements in sifted-key rounds:* Define the set of rounds $I_w = \{k : k \equiv w \pmod{l_c + 1}\}$ and its complement $I_{\bar{w}} = \{k : k \not\equiv w \pmod{l_c + 1}\}$.
>
>    (a) *Bit measurements for rounds in $I_{\bar{w}}$:* For each round $k \in \mathcal{D}_{\text{key}} \cap I_{\bar{w}}$, Alice measures $A_k$ in $\{|0\rangle_{A_k}, |1\rangle_{A_k}\}$ and Bob measures $\{G_0, G_1\}$.
>
>    (b) *Phase-error measurements for rounds in $I_w$:* For each round $k \in \mathcal{D}_{\text{key}} \cap I_w$, Alice measures $A_k$ in $\{|+\rangle_{A_k}, |-\rangle_{A_k}\}$ and Bob measures $\{G_+, G_-\}$. We denote the phase-error rate of the $w$-th partition $e_{\text{ph}}^{(w)}$ as the fraction of events in which their outcomes differ.

To obtain a statistical bound on $e_{\text{ph}}^{(w)}$, we consider a scenario equivalent to the above in which Alice and Bob perform their actions in a different order. First, Alice generates the global state in Eq. (B10), and then Eve applies her global isometry $V'_{T_1^N \to B_1^N E}$, which we can regard as a CPTP map $\Phi'_{T_1^N \to B_1^N}$ by tracing out system $E$. Next, Alice performs all her measurements for the rounds in $I_{\bar{w}}$, learning her setting choice sequence $j_{I_{\bar{w}}}$ for these rounds. Then Alice and Bob perform their measurements for the rounds in $I_w$, learning the value of $e_{\text{ph}}^{(w)}$ and $\vec{n}^{(w)}$ (the restriction of $\vec{n}$ to the rounds in $I_w$). Finally, Bob performs all his measurements for the rounds in $I_{\bar{w}}$. For this reordered scenario, the (unnormalized) state shared by Alice and Bob after Alice measures the rounds in $I_{\bar{w}}$ conditional on obtaining a setting choice sequence $j_{I_{\bar{w}}}$, after tracing out all the systems for the $I_{\bar{w}}$ rounds, is given by:

$$\begin{aligned}
&\operatorname{Tr}_{B_{I_{\bar{w}}}} \left[ \langle j_{I_{\bar{w}}}|_{A_{I_{\bar{w}}}} \Phi'_{T_1^N \to B_1^N} \left( |\Psi'_N\rangle\langle\Psi'_N|_{A_1^N T_1^N} \right) |j_{I_{\bar{w}}}\rangle_{A_{I_{\bar{w}}}} \right] \\
&= \operatorname{Tr}_{B_{I_{\bar{w}}}} \left[ \Phi'_{T_1^N \to B_1^N} \left( \langle j_{I_{\bar{w}}}|_{A_{I_{\bar{w}}}} |\Psi'_N\rangle\langle\Psi'_N|_{A_1^N T_1^N} |j_{I_{\bar{w}}}\rangle_{A_{I_{\bar{w}}}} \right) \right] \\
&= \Pr[j_{I_{\bar{w}}}] \operatorname{Tr}_{B_{I_{\bar{w}}}} \left[ \Phi'_{T_1^N \to B_1^N} \left( |\Psi''_{j_{I_{\bar{w}}}}\rangle\langle\Psi''_{j_{I_{\bar{w}}}}|_{A_{I_w} T_1^N} \right) \right] \\
&= \Pr[j_{I_{\bar{w}}}] \Phi''_{T_1^N \to B_{I_w}} \left( |\Psi''_{j_{I_{\bar{w}}}}\rangle\langle\Psi''_{j_{I_{\bar{w}}}}|_{A_{I_w} T_1^N} \right),
\end{aligned} \quad (B11)$$

where we have defined

$$|\Psi''_{j_{I_{\bar{w}}}}\rangle_{A_{I_w} T_1^N} = \sum_{j_{I_w}} \sqrt{\Pr[j_{I_w}]} |j_{I_w}\rangle_{A_{I_w}} |\Psi'_{j_1^N}\rangle_{T_1^N}, \quad (B12)$$

and

$$\Phi''_{T_1^N \to B_{I_w}} = \operatorname{Tr}_{B_{I_{\bar{w}}}} \circ \Phi'_{T_1^N \to B_1^N}. \quad (B13)$$

Note that in the second equality in Eq. (B11) we have used the fact that $\Pr[j_1^N] = \Pr[j_{I_w}] \Pr[j_{I_{\bar{w}}}]$, since all of Alice's setting choices are independent of one another.

The normalized state conditional on the outcome $j_{I_{\bar{w}}}$ can thus be written as

$$\rho'^{,j_{I_{\bar{w}}}}_{A_{I_w} B_{I_w}} = \Phi''_{T_1^N \to B_{I_w}} \left( |\Psi''_{j_{I_{\bar{w}}}}\rangle\langle\Psi''_{j_{I_{\bar{w}}}}|_{A_{I_w} T_1^N} \right). \quad (B14)$$

Now, consider the measurements performed by Alice and Bob on the rounds in $I_w$, through which they learn the values of $e_{\text{ph}}^{(w)}$ and $\vec{n}^{(w)}$. Again, we can define a simple two-outcome POVM $\{M_{A_{I_w} B_{I_w}}^{\leq,\epsilon,(w)}, M_{A_{I_w} B_{I_w}}^{>,\epsilon,(w)}\}$ that only checks whether $e_{\text{ph}}^{(w)} \leq \mathcal{E}_{\text{ph}}(\vec{n}^{(w)}; N^{(w)}, \epsilon)$ or $e_{\text{ph}}^{(w)} > \mathcal{E}_{\text{ph}}(\vec{n}^{(w)}; N^{(w)}, \epsilon)$. By doing so, we can write

$$\Pr\left[e_{\text{ph}}^{(w)} > \mathcal{E}_{\text{ph}}(\vec{n}^{(w)}; N^{(w)}, \epsilon) \mid j_{I_{\bar{w}}}\right] = \operatorname{Tr}\left[M_{A_{I_w} B_{I_w}}^{>,\epsilon,(w)} \rho'^{,j_{I_{\bar{w}}}}_{A_{I_w} B_{I_w}}\right]. \quad (B15)$$

Next, we show that Eq. (B15) can be rewritten in such a way that it becomes equivalent to Eq. (B9) in the uncorrelated case. Consider the isometry $V^{j_{I_{\bar{w}}}}_{T_{I_w} \to T_1^N}$ defined in the theorem statement for the specific result $j_{I_{\bar{w}}}$, and let us define

$$|\Psi'''_{j_{I_{\bar{w}}}}\rangle_{A_{I_w} T_{I_w}} = (V^{j_{I_{\bar{w}}}}_{T_{I_w} \to T_1^N})^\dagger |\Psi''_{j_{I_{\bar{w}}}}\rangle_{A_{I_w} T_1^N} = \sum_{j_{I_w}} \sqrt{\Pr[j_{I_w}]} |j_{I_w}\rangle_{A_{I_w}} (V^{j_{I_{\bar{w}}}}_{T_{I_w} \to T_1^N})^\dagger |\Psi'_{j_1^N}\rangle_{T_1^N}. \quad (B16)$$

Note that since $\{(V^{j_{I_w}}_{T_{I_w} \to T_1^N})^\dagger |\Psi'_{j_1^N}\rangle_{T_1^N}\}_{j_{I_w}} \in \mathcal{S}_{N^{(w)}}$ by assumption and all states in $\mathcal{S}_{N^{(w)}}$ are normalized, Eq. (B16) is still a valid normalized state. Using Eq. (B16), we can rewrite Eq. (B14) as

$$
\begin{aligned}
\rho^{\prime,j_{I_{\bar w}}}_{A_{I_w} B_{I_w}} &= \Phi''_{T_1^N \to B_{I_w}}\left(|\Psi''_{j_{I_{\bar w}}}\rangle\langle\Psi''_{j_{I_{\bar w}}}|_{A_{I_w} T_1^N}\right) \\
&= \Phi''_{T_1^N \to B_{I_w}}\left(V^{j_{I_{\bar w}}}_{T_{I_w} \to T_1^N}(V^{j_{I_{\bar w}}}_{T_{I_w} \to T_1^N})^\dagger |\Psi''_{j_{I_{\bar w}}}\rangle\langle\Psi''_{j_{I_{\bar w}}}|_{A_{I_w} T_1^N} V^{j_{I_{\bar w}}}_{T_{I_w} \to T_1^N}(V^{j_{I_{\bar w}}}_{T_{I_w} \to T_1^N})^\dagger\right) \\
&= \Phi''_{T_1^N \to B_{I_w}}\left(V^{j_{I_{\bar w}}}_{T_{I_w} \to T_1^N} |\Psi'''_{j_{I_{\bar w}}}\rangle\langle\Psi'''_{j_{I_{\bar w}}}|_{A_{I_w} T_{I_w}} (V^{j_{I_{\bar w}}}_{T_{I_w} \to T_1^N})^\dagger\right) \\
&= \Phi'''_{T_{I_w} \to B_{I_w}}\left(|\Psi'''_{j_{I_{\bar w}}}\rangle\langle\Psi'''_{j_{I_{\bar w}}}|_{A_{I_w} T_{I_w}}\right),
\end{aligned}
\tag{B17}
$$

where we have defined

$$
\Phi'''^{,j_{I_{\bar w}}}_{T_{I_w} \to B_{I_w}}(\sigma) = \Phi''_{T_1^N \to B_{I_w}}\left(V^{j_{I_{\bar w}}}_{T_{I_w} \to T_1^N} \sigma \left(V^{j_{I_{\bar w}}}_{T_{I_w} \to T_1^N}\right)^\dagger\right).
\tag{B18}
$$

Substituting Eq. (B17) into Eq. (B15), we obtain

$$
\Pr\left[e^{(w)}_{\mathbf{ph}} > \mathcal{E}_{\mathrm{ph}}(\vec{n}^{(w)}; N^{(w)}, \epsilon) \mid j_{I_{\bar w}}\right] = \mathrm{Tr}\left[M^{>,\epsilon,(w)}_{A_{I_w} B_{I_w}} \Phi'''^{,j_{I_{\bar w}}}_{T_{I_w} \to B_{I_w}}\left(|\Psi'''_{j_{I_{\bar w}}}\rangle\langle\Psi'''_{j_{I_{\bar w}}}|_{A_{I_w} T_{I_w}}\right)\right].
\tag{B19}
$$

Note that the state $|\Psi'''_{j_{I_{\bar w}}}\rangle_{A_{I_w} T_{I_w}}$ in Eq. (B16) has precisely the form of the uncorrelated source-replacement state in Eq. (B7) for a protocol with $N^{(w)} = |I_w|$ rounds, where the family of states $\{|\Psi_{j_1^N}\rangle_{T_1^N}\}_{j_1^N \in \mathcal{J}^N} \in \mathcal{S}_N$ in Eq. (B7) has been replaced by the family of states $\{(V^{j_{I_{\bar w}}}_{T_{I_w} \to T_1^N})^\dagger |\Psi'_{j_1^N}\rangle_{T_1^N}\}_{j_{I_w} \in \mathcal{J}^{N^{(w)}}} \in \mathcal{S}_{N^{(w)}}$ in Eq. (B16). Because of this, Eq. (B19) has exactly the same form as Eq. (B9) (which is valid for any $N$ and holds for any CPTP map), and thus it follows that

$$
\Pr\left[e^{(w)}_{\mathbf{ph}} > \mathcal{E}_{\mathrm{ph}}(\vec{n}^{(w)}; N^{(w)}, \epsilon) \mid j_{I_{\bar w}}\right] \leq \epsilon.
\tag{B20}
$$

Also, applying the law of total probability, we find that

$$
\begin{aligned}
\Pr\left[e^{(w)}_{\mathbf{ph}} > \mathcal{E}_{\mathrm{ph}}(\vec{n}^{(w)}; N^{(w)}, \epsilon)\right] &= \sum_{j_{I_{\bar w}}} \Pr[j_{I_{\bar w}}] \Pr\left[e^{(w)}_{\mathbf{ph}} > \mathcal{E}_{\mathrm{ph}}(\vec{n}^{(w)}; N^{(w)}, \epsilon) \mid j_{I_{\bar w}}\right] \\
&\leq \sum_{j_{I_{\bar w}}} \Pr[j_{I_{\bar w}}] \epsilon = \epsilon,
\end{aligned}
\tag{B21}
$$

and thus the bound holds even when removing the conditioning on $j_{I_{\bar w}}$.

Note that we have derived this result for the $w$-th phase-error estimation protocol, which we now write explicitly

$$
\Pr_{w\text{-th PEEP}}[e^{(w)}_{\mathbf{ph}} > \mathcal{E}_{\mathrm{ph}}(\vec{n}^{(w)}; N^{(w)}, \epsilon)] \leq \epsilon.
\tag{B22}
$$

However, what we want is to bound the phase-error rate for the original *Phase-error estimation protocol* defined in Appendix A. To connect the two, let us define:

---

**Full phase-error estimation protocol (PEEP):**

1-5. Same as in **Actual protocol (source replaced)**

6. *Measurements in sifted-key rounds:* Partition the rounds $k \in \{1, ..., N\}$ into the sets $I_w = \{k : k \equiv w \pmod{l_c + 1}\}$. Then:

    (a) *Phase-error measurements for rounds in $I_w$:* For each round $k \in \mathcal{D}_{\mathrm{key}} \cap I_w$, Alice measures $A_k$ in $\{|+\rangle_{A_k}, |-\rangle_{A_k}\}$ and Bob measures $\{G_+, G_-\}$. We denote the phase-error rate of the $w$-th partition $e^{(w)}_{\mathbf{ph}}$ as the fraction of events in which their outcomes differ.

---

Note that this scenario is essentially identical to the original *Phase-error estimation protocol* defined in Appendix A. The only difference is that the phase-error rate is tracked separately for different $I_w$. However, we can define the overall phase-error rate as

$$e_{\mathbf{ph}} = \frac{\sum_{w=0}^{l_c} n_K^{(w)} e_{\mathbf{ph}}^{(w)}}{n_K}, \tag{B23}$$

where $n_K^{(w)}$ is the number of sifted key bits from rounds in $I_w$ and $n_K = \sum_{w=0}^{l_c} n_K^{(w)}$ is the total number of sifted key bits. Importantly, the statistics of $e_{\mathbf{ph}}$ in this scenario must be identical as in the original *Phase-error estimation protocol* defined in Appendix A. Moreover, note that, for all $w \in \{0, 1, \ldots, l_c\}$,

$$\Pr_{\text{Full PEEP}}[e_{\mathbf{ph}}^{(w)} > \mathcal{E}_{\text{ph}}(\vec{n}^{(w)}; N^{(w)}, \epsilon)] = \Pr_{w\text{-th PEEP}}[e_{\mathbf{ph}}^{(w)} > \mathcal{E}_{\text{ph}}(\vec{n}^{(w)}; N^{(w)}, \epsilon)] \leq \epsilon. \tag{B24}$$

This is because the statistics of the random variables $e_{\mathbf{ph}}^{(w)}$ and $\vec{n}^{(w)}$ depend only on the marginal state on systems $A_{I_w} B_{I_w}$ and on the measurements performed in these rounds. Since the Full PEEP and the $w$-th PEEP only differ in the measurements performed on the complementary systems $A_{I_{\bar{w}}} B_{I_{\bar{w}}}$, and these measurements do not affect the marginal state on $A_{I_w} B_{I_w}$ (since Alice and Bob never perform any operation on systems $A_{I_w} B_{I_w}$ depending on the outcome of the measurements on $A_{I_{\bar{w}}} B_{I_{\bar{w}}}$), the *a priori* distribution of $(e_{\mathbf{ph}}^{(w)}, \vec{n}^{(w)})$ must be identical in both scenarios. Now, for the *Full PEEP*, define the following event specified by a relationship between random variables,

$$\mathcal{B} = \left\{ e_{\mathbf{ph}} > \frac{\sum_{w=0}^{l_c} n_K^{(w)} \mathcal{E}_{\text{ph}}(\vec{n}^{(w)}; N^{(w)}, \epsilon)}{n_K} \right\}. \tag{B25}$$

From the definition of $e_{\mathbf{ph}}$ in Eq. (B23), for this event to occur, we must have

$$\sum_{w=0}^{l_c} n_K^{(w)} e_{\mathbf{ph}}^{(w)} > \sum_{w=0}^{l_c} n_K^{(w)} \mathcal{E}_{\text{ph}}(\vec{n}^{(w)}; N^{(w)}, \epsilon). \tag{B26}$$

This implies that for at least one value of $w \in \{0, 1, \ldots, l_c\}$, we must have $e_{\mathbf{ph}}^{(w)} > \mathcal{E}_{\text{ph}}(\vec{n}^{(w)}; N^{(w)}, \epsilon)$. Otherwise, if $e_{\mathbf{ph}}^{(w)} \leq \mathcal{E}_{\text{ph}}(\vec{n}^{(w)}; N^{(w)}, \epsilon)$ for all $w$, then

$$\sum_{w=0}^{l_c} n_K^{(w)} e_{\mathbf{ph}}^{(w)} \leq \sum_{w=0}^{l_c} n_K^{(w)} \mathcal{E}_{\text{ph}}(\vec{n}^{(w)}; N^{(w)}, \epsilon), \tag{B27}$$

which would contradict the occurrence of event $\mathcal{B}$. In other words, we have that[7]

$$\mathcal{B} \subseteq \bigcup_{w=0}^{l_c} \left\{ e_{\mathbf{ph}}^{(w)} > \mathcal{E}_{\text{ph}}(\vec{n}^{(w)}; N^{(w)}, \epsilon) \right\}, \tag{B28}$$

and applying the union bound, we obtain

$$\Pr[\mathcal{B}] \leq \sum_{w=0}^{l_c} \Pr\left[ e_{\mathbf{ph}}^{(w)} > \mathcal{E}_{\text{ph}}(\vec{n}^{(w)}; N^{(w)}, \epsilon) \right] \leq \sum_{w=0}^{l_c} \epsilon = (l_c + 1)\epsilon, \tag{B29}$$

where the second inequality follows from Eq. (B24) for each $w$. Therefore, in the *Full PEEP* (and thus also in the original *Phase-error estimation protocol* defined in Appendix A),

$$\Pr\left[ e_{\mathbf{ph}} > \frac{\sum_{w=0}^{l_c} n_K^{(w)} \mathcal{E}_{\text{ph}}(\vec{n}^{(w)}; N^{(w)}, \epsilon)}{n_K} \right] \leq (l_c + 1)\epsilon, \tag{B30}$$

as we wanted to prove. □

---

[7] Note that, if $n_K^{(w)} = 0$, $e_{\mathbf{ph}}^{(w)}$ is not well-defined. Here, we are trivially considering that if $n_K^{(w)} = 0$, we set $e_{\mathbf{ph}}^{(w)} := 0$ and $\mathcal{E}_{\text{ph}}(\vec{n}^{(w)}; N^{(w)}, \epsilon) := 0$.

**Corollary 1** (Per-round admissibility conditions). *Consider a prepare-and-measure QKD protocol with an uncorrelated source, where Alice emits a state $|\psi_{j_k}^{(k)}\rangle_{T_k}$ when choosing setting $j_k \in \mathcal{J}$ in round $k$. Suppose there exists an admissibility set $\mathcal{S}$ of state families indexed by $j \in \mathcal{J}$ such that the phase-error bound in Eq. (3) holds as long as*

$$\{|\psi_{j_k}^{(k)}\rangle\}_{j_k \in \mathcal{J}} \in \mathcal{S}, \quad \forall k. \tag{B31}$$

*Now consider the analogous protocol with a source exhibiting correlations up to length $l_c$. For any sequence of settings $j_{k-l_c}^{k+l_c}$ (interpreted with the boundary conventions in Remark 2 below), define the joint state emitted in rounds $k$ to $k+l_c$ as*

$$|\Psi_{j_{k-l_c}^{k+l_c}}^{(k)}\rangle_{T_k^{k+l_c}} = \bigotimes_{l=k}^{k+l_c} |\psi_{j_{l-l_c}^{l}}^{(l)}\rangle_{T_l}. \tag{B32}$$

*Suppose that, for every round $k$ and every fixed choice of past and future settings $(j_{k-l_c}^{k-1}, j_{k+1}^{k+l_c})$, the family $\{|\Psi_{j_{k-l_c}^{k+l_c}}^{(k)}\rangle_{T_k^{k+l_c}}\}_{j_k \in \mathcal{J}}$ has the same Gram matrix as an acceptable family of single-round states $\{|\varphi_j^{(j_{k-l_c}^{k-1}, j_{k+1}^{k+l_c})}\rangle_{T_k}\}_{j \in \mathcal{J}} \in \mathcal{S}$. Equivalently, there exists an isometry*

$$V_{T_k \to T_k^{k+l_c}}^{(j_{k-l_c}^{k-1}, j_{k+1}^{k+l_c})} : \mathcal{H}_{T_k} \to \mathcal{H}_{T_k^{k+l_c}}, \tag{B33}$$

*such that*

$$\left\{ \left(V_{T_k \to T_k^{k+l_c}}^{(j_{k-l_c}^{k-1}, j_{k+1}^{k+l_c})}\right)^\dagger |\Psi_{j_{k-l_c}^{k+l_c}}^{(k)}\rangle_{T_k^{k+l_c}} \right\}_{j_k \in \mathcal{J}} \in \mathcal{S}. \tag{B34}$$

*Then, partitioning the rounds $\{1, ..., N\}$ into $(l_c + 1)$ sets $I_w = \{k : k \equiv w \bmod (l_c+1)\}$ with $w = 0, ..., l_c$, the phase-error rate in the correlated scenario satisfies Eq. (B6).*

**Remark 2.** *In the statement of Corollary 1, we use the shorthand $j_{k-l_c}^{k+l_c} = j_{k-l_c} j_{k-l_c+1} ... j_{k+l_c}$ to denote setting sequences. For rounds near the boundaries of the protocol (i.e., when $k - l_c < 1$ or $k + l_c > N$), these indices should be understood as appropriately truncated to the range $[1, N]$, i.e., $j_{k-l_c}^{k+l_c} \equiv j_{\max(1, k-l_c)}^{\min(k+l_c, N)}$. The same is true for the sequence of systems $T_k^{k+l_c} \equiv T_k^{\min(k+l_c, N)}$. The proof of Corollary 1 below handles these boundary cases explicitly.*

*Proof.* We simply need to verify that the conditions of Theorem 2 are satisfied. Specifically, we must show that, for each $w \in \{0, 1, ..., l_c\}$, there exists a global isometry such that

$$\left\{ \left(V_{T_{I_w} \to T_1^N}^{j_{I_{\bar{w}}}}\right)^\dagger |\Psi'_{j_1^N}\rangle_{T_1^N} \right\}_{j_{I_w} \in \mathcal{J}^{N(w)}} \in \mathcal{S}_{N(w)}, \tag{B35}$$

where $\mathcal{S}_{N(w)}$ is defined as:

$$\mathcal{S}_{N(w)} = \left\{ \left\{ \bigotimes_{k \in I_w} |\varphi_{j_k}^{(k)}\rangle_{T_k} \right\}_{j_{I_w} \in \mathcal{J}^{N(w)}} : \left\{ |\varphi_{j_k}^{(k)}\rangle_{T_k} \right\}_{j_k \in \mathcal{J}} \in \mathcal{S}, \forall k \in I_w \right\}. \tag{B36}$$

Let $k_{\min}^{(w)} = \min I_w$ denote the smallest round index in partition $I_w$. Note that $k_{\min}^{(w)} = w$ for $w \in \{1, \ldots, l_c\}$ and $k_{\min}^{(0)} = l_c + 1$. Since consecutive elements of $I_w$ differ by exactly $l_c + 1$, the blocks $\{k, k+1, \ldots, \min(k+l_c, N)\}$ for $k \in I_w$ partition the rounds $\{k_{\min}^{(w)}, \ldots, N\}$. We can therefore write the global emitted state as

$$|\Psi'_{j_1^N}\rangle_{T_1^N} = \bigotimes_{k=1}^{N} |\psi_{j_{\max(1,k-l_c)}^k}^{(k)}\rangle_{T_k} = \left( \bigotimes_{k=1}^{k_{\min}^{(w)}-1} |\psi_{j_{\max(1,k-l_c)}^k}^{(k)}\rangle_{T_k} \right) \otimes \left( \bigotimes_{k \in I_w} \bigotimes_{m=k}^{\min(k+l_c, N)} |\psi_{j_{\max(1,m-l_c)}^m}^{(m)}\rangle_{T_m} \right)$$
$$= \left( \bigotimes_{k=1}^{k_{\min}^{(w)}-1} |\psi_{j_{\max(1,k-l_c)}^k}^{(k)}\rangle_{T_k} \right) \otimes \left( \bigotimes_{k \in I_w} |\Psi_{j_{\max(1,k-l_c)}^{\min(k+l_c, N)}}^{(k)}\rangle_{T_k^{\min(k+l_c, N)}} \right). \tag{B37}$$

Now, define the global isometry for group $I_w$ conditional on $j_{I_{\bar{w}}}$ as:

$$V_{T_{I_w} \to T_1^N}^{j_{I_{\bar{w}}}} = \left( \bigotimes_{k=1}^{k_{\min}^{(w)}-1} |\psi_{j_{\max(1,k-l_c)}^k}^{(k)}\rangle_{T_k} \right) \otimes \left( \bigotimes_{k \in I_w} V_{T_k \to T_k^{\min(k+l_c,N)}}^{(j_{\max(1,k-l_c)}^{k-1}, j_{k+1}^{\min(k+l_c,N)})} \right). \quad (B38)$$

This is well-defined because the per-round isometries act on disjoint subsystems. Moreover, for each $k \in I_w$, the indices $(j_{\max(1,k-l_c)}^{k-1}, j_{k+1}^{\min(k+l_c,N)})$ lie entirely in $I_{\bar{w}}$. Similarly, for each $k \in \{1, \ldots, k_{\min}^{(w)}-1\}$, all indices in $\{\max(1, k-l_c), \ldots, k\}$ are strictly less than $k_{\min}^{(w)} = \min I_w$ and hence belong to $I_{\bar{w}}$. Therefore $V_{T_{I_w} \to T_1^N}^{j_{I_{\bar{w}}}}$ depends only on $j_{I_{\bar{w}}}$, as required by Theorem 2.

Applying the adjoint of the isometry to the global emitted state, the first tensor factors contribute $\prod_{k=1}^{k_{\min}^{(w)}-1} \langle \psi_{j_{\max(1,k-l_c)}^k}^{(k)} | \psi_{j_{\max(1,k-l_c)}^k}^{(k)} \rangle = 1$, and we obtain

$$(V_{T_{I_w} \to T_1^N}^{j_{I_{\bar{w}}}})^\dagger |\Psi'_{j_1^N}\rangle_{T_1^N} = \bigotimes_{k \in I_w} \left( V_{T_k \to T_k^{\min(k+l_c,N)}}^{(j_{\max(1,k-l_c)}^{k-1}, j_{k+1}^{\min(k+l_c,N)})} \right)^\dagger |\Psi_{j_{\max(1,k-l_c)}^{\min(k+l_c,N)}}^{(k)}\rangle_{T_k^{\min(k+l_c,N)}}. \quad (B39)$$

By assumption, for each $k \in I_w$, the family

$$\left\{ \left( V_{T_k \to T_k^{\min(k+l_c,N)}}^{(j_{\max(1,k-l_c)}^{k-1}, j_{k+1}^{\min(k+l_c,N)})} \right)^\dagger |\Psi_{j_{\max(1,k-l_c)}^{\min(k+l_c,N)}}^{(k)}\rangle_{T_k^{\min(k+l_c,N)}} \right\}_{j_k \in \mathcal{J}} \in \mathcal{S}, \quad (B40)$$

and therefore the family $\{(V_{T_{I_w} \to T_1^N}^{j_{I_{\bar{w}}}})^\dagger |\Psi'_{j_1^N}\rangle_{T_1^N}\}_{j_{I_w} \in \mathcal{J}^{N^{(w)}}} \in \mathcal{S}_{N^{(w)}}$, since it has the product form required by Eq. (B36), as we wanted to prove. $\square$

**Corollary 2** (Fidelity bound to reference states). *Consider a prepare-and-measure QKD protocol with an uncorrelated source, and suppose there exists a set of reference states $\{|\phi_j\rangle\}_{j \in \mathcal{J}}$ such that the phase-error bound in Eq. (B2) holds as long as*

$$\left| \langle \phi_{j_k} | \psi_{j_k}^{(k)} \rangle_{T_k} \right|^2 \geq 1 - \xi_{j_k}, \quad \forall k, \forall j_k \in \mathcal{J}. \quad (B41)$$

*For an analogous protocol with a source with correlations up to length $l_c$, suppose that for every round $k$ and every choice of past and future settings $(j_{k-l_c}^{k-1}, j_{k+1}^{k+l_c})$, there exists a family of states $\{|\phi_{j_{k-l_c}^{k+l_c}}^{(k)}\rangle_{T_k}\}_{j_k \in \mathcal{J}}$ with the same Gram matrix as the family of reference states $\{|\phi_j\rangle\}_j$, and that there exist states $|\lambda_{j_{k-l_c}^{k-1}, j_{k+1}^{k+l_c}}^{(k)}\rangle_{T_{k+1}^{k+l_c}}$ (independent of $j_k$) such that*

$$\left| \langle \phi_{j_{k-l_c}^{k+l_c}}^{(k)} |_{T_k} \otimes \langle \lambda_{j_{k-l_c}^{k-1}, j_{k+1}^{k+l_c}}^{(k)} |_{T_{k+1}^{k+l_c}} |\Psi_{j_{k-l_c}^{k+l_c}}^{(k)}\rangle_{T_k^{k+l_c}} \right|^2 \geq 1 - \xi_{j_k}, \quad \forall j_k. \quad (B42)$$

*Then, the phase-error rate bound in Eq. (B6) holds for this correlated scenario.*

*Proof.* The proof for the uncorrelated case defines the per-round admissibility set

$$\mathcal{S} = \left\{ \{|\varphi_j\rangle\}_{j \in \mathcal{J}} : |\langle \phi_j | \varphi_j \rangle|^2 \geq 1 - \xi_j, \ \forall j \in \mathcal{J} \right\}. \quad (B43)$$

To apply Corollary 1, we need to prove that for any fixed $(j_{k-l_c}^{k-1}, j_{k+1}^{k+l_c})$, there exists an isometry such that

$$\left\{ \left( V_{T_k \to T_k^{k+l_c}}^{(j_{k-l_c}^{k-1}, j_{k+1}^{k+l_c})} \right)^\dagger |\Psi_{j_{k-l_c}^{k+l_c}}^{(k)}\rangle_{T_k^{k+l_c}} \right\}_{j_k \in \mathcal{J}} \in \mathcal{S}. \quad (B44)$$

Note that, since the families $\{|\phi_{j_{k-l_c}^{k+l_c}}^{(k)}\rangle_{T_k}\}_{j_k \in \mathcal{J}}$ and $\{|\phi_{j_k}\rangle_{T_k}\}_{j_k \in \mathcal{J}}$ have the same Gram matrix by assumption, there must exist a unitary operation depending on $(j_{k-l_c}^{k-1}, j_{k+1}^{k+l_c})$ that takes the latter family to the former. Moreover, trivially, there exists an isometry $T_k \to T_k^{k+l_c}$ depending on $(j_{k-l_c}^{k-1}, j_{k+1}^{k+l_c})$ that takes an arbitrary state $|\cdot\rangle_{T_k}$ to

$|\cdot\rangle_{T_k} \,|\lambda^{(k)}_{j^{k-1}_{k-l_c},j^{k+l_c}_{k+1}}\rangle_{T^{k+l_c}_{k+1}}$. Combining these two, we obtain that, for each fixed $(j^{k-1}_{k-l_c}, j^{k+l_c}_{k+1})$, there exists an isometry such that

$$V^{(j^{k-1}_{k-l_c}, j^{k+l_c}_{k+1})}_{T_k \to T^{k+l_c}_k} |\phi_{j_k}\rangle_{T_k} = |\phi^{(k)}_{j^{k+l_c}_{k-l_c}}\rangle_{T_k} \,|\lambda^{(k)}_{j^{k-1}_{k-l_c},j^{k+l_c}_{k+1}}\rangle_{T^{k+l_c}_{k+1}}, \quad \forall j_k. \tag{B45}$$

Moreover, note that

$$\left| \langle \phi_{j_k}|_{T_k} \left( V^{(j^{k-1}_{k-l_c}, j^{k+l_c}_{k+1})}_{T_k \to T^{k+l_c}_k} \right)^\dagger |\Psi^{(k)}_{j^{k+l_c}_{k-l_c}}\rangle_{T^{k+l_c}_k} \right|^2 = \left| \langle \phi^{(k)}_{j^{k+l_c}_{k-l_c}}|_{T_k} \otimes \langle \lambda^{(k)}_{j^{k-1}_{k-l_c},j^{k+l_c}_{k+1}}|_{T^{k+l_c}_{k+1}} \,|\Psi^{(k)}_{j^{k+l_c}_{k-l_c}}\rangle_{T^{k+l_c}_k} \right|^2 \geq 1 - \xi_{j_k}, \tag{B46}$$

where we have used Eq. (B42). This implies Eq. (B44), as we wanted to prove.

□

**Lemma 1** (Unbounded correlations). *Consider a prepare-and-measure QKD protocol with a source exhibiting correlations of unbounded length, and let $|\Psi^{(\infty)}_N\rangle_{A^N_1 T^N_1}$ be the source-replacement state for this source. Also, let $|\Psi^{(l_c)}_N\rangle_{A^N_1 T^N_1}$ be the source-replacement state for a fictitious source with correlations up to length $l_c$. Suppose that the trace distance between these two states satisfies*

$$T\Big( |\Psi^{(\infty)}_N\rangle\langle\Psi^{(\infty)}_N|_{A^N_1 T^N_1}, |\Psi^{(l_c)}_N\rangle\langle\Psi^{(l_c)}_N|_{A^N_1 T^N_1} \Big) \leq d, \tag{B47}$$

*and that, if Alice were to prepare $|\Psi^{(l_c)}_N\rangle_{A^N_1 T^N_1}$, then the following phase-error rate bound holds for any eavesdropping attack*

$$\Pr_{(l_c)}[e_{\mathbf{ph}} > \mathcal{E}_{\mathrm{ph}}(\vec{n}; N, \epsilon)] \leq \epsilon. \tag{B48}$$

*Then, if Alice prepares $|\Psi^{(\infty)}_N\rangle_{A^N_1 T^N_1}$, the following phase-error bound holds for any eavesdropping attack*

$$\Pr_{(\infty)}[e_{\mathbf{ph}} > \mathcal{E}_{\mathrm{ph}}(\vec{n}; N, \epsilon)] \leq \epsilon + d. \tag{B49}$$

*Proof.* Consider a fixed attack by Eve, which can be described as a CPTP map $\Phi_{T^N_1 \to B^N_1}$. Using the same reasoning as in the beginning of the proof of Theorem 2, we can express the failure probability of the phase-error rate bound for each source-replacement state as

$$\Pr_{(l_c)}[e_{\mathbf{ph}} > \mathcal{E}_{\mathrm{ph}}(\vec{n}; N, \epsilon)] = \mathrm{Tr}\Big[ M^{>,\epsilon}_{A^N_1 B^N_1} \Phi_{T^N_1 \to B^N_1}\Big( |\Psi^{(l_c)}_N\rangle\langle\Psi^{(l_c)}_N|_{A^N_1 T^N_1} \Big) \Big]. \tag{B50}$$

and

$$\Pr_{(\infty)}[e_{\mathbf{ph}} > \mathcal{E}_{\mathrm{ph}}(\vec{n}; N, \epsilon)] = \mathrm{Tr}\Big[ M^{>,\epsilon}_{A^N_1 B^N_1} \Phi_{T^N_1 \to B^N_1}\Big( |\Psi^{(\infty)}_N\rangle\langle\Psi^{(\infty)}_N|_{A^N_1 T^N_1} \Big) \Big]. \tag{B51}$$

where $M^{>,\epsilon}_{A^N_1 B^N_1}$ is a POVM element.

Since the trace distance $T(\rho, \sigma) := \frac{1}{2}\|\rho - \sigma\|_1$ is non-increasing under CPTP maps, we have that

$$\begin{aligned} &T\bigg( \Phi_{T^N_1 \to B^N_1}\Big( |\Psi^{(\infty)}_N\rangle\langle\Psi^{(\infty)}_N|_{A^N_1 T^N_1} \Big), \Phi_{T^N_1 \to B^N_1}\Big( |\Psi^{(l_c)}_N\rangle\langle\Psi^{(l_c)}_N|_{A^N_1 T^N_1} \Big) \bigg) \\ &\leq T\Big( |\Psi^{(\infty)}_N\rangle\langle\Psi^{(\infty)}_N|_{A^N_1 T^N_1}, |\Psi^{(l_c)}_N\rangle\langle\Psi^{(l_c)}_N|_{A^N_1 T^N_1} \Big) \leq d. \end{aligned} \tag{B52}$$

Moreover, for any POVM element $0 \leq M \leq \mathbb{I}$, $\big|\mathrm{Tr}[M(\rho - \sigma)]\big| \leq T(\rho, \sigma)$. Therefore, we must have that

$$\Pr_{(\infty)}[e_{\mathbf{ph}} > \mathcal{E}_{\mathrm{ph}}(\vec{n}; N, \epsilon)] \leq \Pr_{(l_c)}[e_{\mathbf{ph}} > \mathcal{E}_{\mathrm{ph}}(\vec{n}; N, \epsilon)] + d \leq \epsilon + d, \tag{B53}$$

for the fixed CPTP map $\Phi_{T^N_1 \to B^N_1}$. Since by assumption Eq. (B48) holds for any CPTP map $\Phi_{T^N_1 \to B^N_1}$, then Eq. (B49) must also hold for any CPTP map.

□

19**Remark 3** (Interference-based protocols). *All our previous results extend naturally to interference-based protocols (also known as MDI-type protocols), i.e., protocols in which Alice and Bob send quantum states to an untrusted middle node Charlie and classical announcements from Charlie determine the detected rounds. Concretely, one can define a source-replaced version of the actual protocol and an associated phase-error estimation protocol analogously to Appendix A; see, e.g., [38] for a general formulation.*

*For such protocols, our framework can incorporate encoding correlations in both Alice's and Bob's transmitter. To apply our framework, one should define $\mathcal{J} = \mathcal{J}_A \times \mathcal{J}_B$ (the alphabet of setting combinations for both users), $j_k = (j_{A,k}, j_{B,k})$ (the joint setting in round $k$) and $T_k = T_{A,k} T_{B,k}$ (the two optical systems emitted by Alice and Bob in round $k$). Then $|\psi_{j_1^k}^{(k)}\rangle_{T_k}$ denotes the joint state emitted by Alice and Bob given the joint setting history $j_1^k$.*

*The original (uncorrelated) security proof for the interference-based protocol should specify an admissibility set $\mathcal{S}_N$ (resp. $\mathcal{S}$) for the joint emitted states, typically including the tensor-product structure constraint $|\psi_{j_k}^{(k)}\rangle_{T_k} = |\psi_{j_{A,k}}^{A,(k)}\rangle_{T_{A,k}} \otimes |\psi_{j_{B,k}}^{B,(k)}\rangle_{T_{B,k}}$. With these identifications, the statements of Theorem 2, Corollaries 1 and 2, and Lemma 1 apply verbatim.*


	 \fi

        \end{document}